%% file: paper_draft.tex
\documentclass[journal, twoside]{IEEEtran}

\usepackage{amsfonts,subfigure,epsfig,psfrag,enumerate}
\usepackage{color,times,amsfonts,amsmath,amssymb,cite,verbatim}
\usepackage{amsmath, ifthen}
\usepackage[active]{srcltx}  %SRC Specials for DVI search
\usepackage{subfigure,epsfig,psfrag,enumerate}
\usepackage{tikz}
\usetikzlibrary{arrows,snakes,backgrounds}
\usepackage{algpseudocode} 
\usepackage{algorithm} 
\usepackage{multirow}
\usepackage{setspace}
\usepackage{hyperref}

\usepackage{pgfplots}
\usetikzlibrary{plotmarks}

\newcommand{\C}[1]{ \ensuremath{ C \left(#1\right) } }
\newcommand{\bigO}[1]{ \ensuremath{ \mathcal{O} \left(#1\right) } }
\newcommand{\abs}[1]{\ensuremath{\left|#1\right|}}
\newcommand{\norm}[1]{\ensuremath{\left\|#1\right\|}}
\newcommand{\set}[1]{\ensuremath{\left\{#1\right\}}}

\newcommand{\figwidth}{7cm}
\newcommand{\figscale}{1}
\newcommand{\herm}[1]{\ensuremath{#1^*}}

\newcommand{\pn}{ \ensuremath{ \sigma_N^2} }
\newcommand{\nw}{ \ensuremath{ N_w} }

\newcommand{\pji}{ \ensuremath{ p_{j\rightarrow i}} }
 
\newcommand{\UE}{MD} 
\newcommand{\UEs}{MDs}

\newtheorem{thm}{Theorem}%[section]
\newtheorem{cor}{Corollary}
\newtheorem{lem}{Lemma}
\newtheorem{prop}{Proposition}

\newtheorem{rem}{Remark}%[section]

\newlength\figureheight 
\newlength\figurewidth

\def\addlegendimage{\csname pgfplots@addlegendimage\endcsname}

\pgfplotsset{compat=newest}
\title{An Efficient Clustering Algorithm for Device-to-Device Assisted Virtual MIMO}
\author{S. Hossein Seyedmehdi and  Gary Boudreau%\\ hossein@comm.utoronto.ca,  gary.boudreau@ericsson.com
\thanks{S. H. Seyedmehdi is with the Department of Electrical and Computer Engineering, University of Toronto, Toronto, ON, M5S 3G4, Canada (email: hossein@comm.utoronto.ca). G. Boudreau is with Ericsson, Ottawa, ON, Canada (email: gary.boudreau@ericsson.com).} 
\thanks{This work was supported in part by Ericsson Canada and the Natural Sciences and Engineering Research Council (NSERC) of Canada.}
}
\begin{document}

%\markboth{Journal of ...}
%{Seyedmehdi \MakeLowercase{\textit{et al.}}: An Efficient Clustering Algorithm for Device-to-Device Assisted Virtual MIMO}

\maketitle
% \tableofcontents
% \listoftables
% \listoffigures

\begin{abstract}
In this paper, the utilization of mobile devices (MDs) as decode-and-forward relays in a device-to-device assisted virtual MIMO (VMIMO) system is studied. Single antenna \UEs~are randomly distributed on a 2D plane according to a Poisson point process, and only a subset of them are sources leaving other idle \UEs~available to assist them (relays). Our goal is to develop an efficient algorithm to cluster each source with a subset of available relays to form a VMIMO system under a limited feedback assumption. We first show that the NP-hard optimization problem of precoding in our scenario  can be approximately solved by semidefinite relaxation. We investigate a special case with a single source and analytically derive an upper bound on the average spectral efficiency of the VMIMO system. Then, we propose an optimal greedy algorithm that achieves this bound. We further exploit these results to obtain a polynomial time clustering algorithm for the general case with multiple sources. Finally, numerical simulations are performed to compare the performance of our algorithm with that of an exhaustive clustering algorithm, and it shown that  these numerical results corroborate the efficiency of our algorithm. 
\end{abstract}


\section{Introduction}
 
\IEEEPARstart{T}{he} use of multiple transmit and receive antennas (MIMO) has been perceived as a promising technique to enhance the spectral efficiency of wireless systems. In practice, due to the size limitation, only one transmit  antenna can usually fit inside each mobile device (\UE) especially in low cost legacy devices. However, multiple single-antenna \UEs~can be clustered together to create a virtual MIMO (VMIMO) system. In most current wireless standards, including the long term evolution (LTE) and IEEE 802.11, \UEs~that are scheduled to transmit comprise only a small fraction of all present \UEs, thus leaving numerous idle \UEs~available to form device-to-device assisted VMIMO systems. One of the new challenges that arises in this realm is that given a large number of \UEs, a proper subset of them must be selected to form the VMIMO system (also known as the \emph{clustering} or  \emph{user grouping} problem).  

Utilizing  a cooperating device has been shown to improve the spectral efficiency and diversity while alleviating the outage behaviour \cite{Laneman2003, Laneman2004, Sendonaris2003a}. This cooperation is performed by a second wireless device (the relay) relaying the message of a first wireless device (the source). 
%In the seminal work of Laneman et al. \cite{Laneman2003}, the diversity gain achieved by this cooperation is investigated, and it is shown that full diversity can be achieved for the repetition-based decode-and-forward (DF) scheme. These results have inspired numerous works on selection algorithms in the cooperative diversity area. 

There is a rich body of literature on how to select the proper relay(s) to cooperate with the source(s) that can be categorized into three trends: In  the earlier works, e.g., \cite{Ibrhim2008, Bletsas2006, Adinoyi2009, Zhao2006}, selection of a \emph{single} best relay (for a single source) based on different performance metrics including the best instantaneous channel condition \cite{Bletsas2006} or the trade off between the amplify-and-forward  (AF) error probability and the power consumption \cite{Zhao2006} is studied. % For instance, in \cite{Ibrhim2008}, one DF relay is selected when the direct link between the source and destination has low quality.%, and it is shown that full diversity is guaranteed while a higher spectral efficiency can be achieved. 
 Later works have generalized the single relay selection concept to \emph{multiple} relay selection in an effort to find  optimal relay selection methods under various assumptions~\cite{Michalopoulos2006, Jing2009, Ikki2009, Hwang2009, Guo2010, Shah2009, Vakil2008, Nam2008, Yu2012}. Most of these works consider AF relaying (e.g., \cite{Michalopoulos2006, Jing2009, Ikki2009, Hwang2009, Guo2010}). The  DF two-hop relaying technique is considered in \cite{Vakil2008, Nam2008, Yu2012}. %To mitigate the interference, in  \cite{Vakil2008}, only the source and relays that reside in the same \emph{cooperation region} are allowed to be grouped together, and the effect of the cooperation region size on the network sum-rate is investigated. 
Both  \cite{Nam2008} and  \cite{Yu2012} investigate algorithms to select $m$ best relays out of $N$ uniformly distributed available relays for a single source, and obtain  approximations on the cumulative distribution function (CDF) of the source spectral efficiency when $m$ relays are selected. %Channel gains considered in \cite{Nam2008} are modelled by only Rayleigh fading whereas \cite{Yu2012} assumes that mean channel gains are known. 
In recent works, relay selection for multiple sources is considered in \cite{Wu2012, Beres2008, Yu2012a, Zhou2013, Li2010}.  AF relaying is considered in \cite{Wu2012}  where $L$ relays are selected. DF selection is considered in  \cite{Beres2008, Yu2012a, Zhou2013, Li2010}.  In \cite{Beres2008} and \cite{Yu2012a}, orthogonal channels are assumed for each source, thus eliminating the effect of the leakage interference. In \cite{Beres2008} and \cite{Zhou2013} a fixed number of relays (one) is selected for each source.  The relay selection problem is approached from a \emph{pricing based game model} in \cite{Li2010, Wang2009}. In~\cite{Li2010} each relay demands a price (virtual currency) for cooperation, and each source bids a price to recruit relays, and in addition, a \emph{competitive price adjustment process} is discussed. Although selection is not a new concept, a scarcity of studies on clustering (grouping) algorithms  from a MIMO perspective is apparent. For instance, most previous works including \cite{Wu2012, Beres2008, Yu2012a, Zhou2013, Li2010}  consider only single antenna at the receiver and fail to examine the effect of the precoding.

%In a related line of work, clustering algorithms have long been studied in the context of wireless sensor networks (WSNs), e.g., the pioneering works \cite{Younis2004, Bandyopadhyay2003, Amis2000}. These studies are mostly concerned with determining the best cluster head such that firstly, the power consumption is minimized, and secondly, the multi-hop link latency from a sensor to the cluster head is bounded. Given a predetermined  clustering, \cite{Liang2009} proposes two algorithms (based on minimum spanning tree and singular value decomposition) to select the best \emph{channels} between two clusters of sensors in a multi-hop VMIMO setting.  Nevertheless, the clustering algorithms in WSNs are often  aimed at maximizing the network lifetime whereas in the cellular communication systems, the diversity is exploited to maximize the spectral efficiency, and this difference  leads to substantially incomparable algorithms.

Recently, the problem of clustering algorithms for VMIMO has gained  increasing attention since VMIMO is recognized  as a promising future trend of  communication systems \cite{Kurve2009, Bhat2012, Yu2011}. %In particular, there has been a considerable number of studies on VMIMO relevant to the LTE standard both from a macrocell perspective \cite{%Badic2012, Wang2011,  Bhat2012} as well as a femtocell perspective~\cite{Yu2011}.   
The problem of joint grouping and precoding for fixed size VMIMO where precoding weights are continuous is studied in \cite{Hong2013}. 
The uplink pair selection problem (e.g., \cite{%Chen2008, Li2008, Ruder2011, 
Rui2012}) is extended to multi user uplink grouping for a single destination in~\cite{Karimi2013}. In this work a fixed number of users are selected (and grouped together) such that a proportional fairness utility  is maximized. Nonetheless, the cooperation between  nodes is not considered, and each node transmits its individual message.

In this paper, we treat the problem of  clustering algorithms for multiple sources by utilizing idle \UEs~as assisting devices (relays) in a VMIMO setup with limited feedback. In the scenario under consideration, all \UEs~share the same channel, and they are spatially distributed according to a Poisson point process. A source either transmits its message without forming a VMIMO system, or it can be clustered with other idle \UEs~in a  VMIMO setup if the latter act improves its spectral efficiency. When a source participates in the VMIMO configuration, its message is conveyed to its serving  access point (AP) in a two-phase DF relaying manner, and unlike some prior studies (e.g., \cite{Nam2008, Yu2012}), the direct link is also considered. Moreover, an approximate precoding is performed by which the transmit signal of each \UE~in the second phase is multiplied by a discrete unity complex number. Due to the limited feedback assumption, unlike \cite{Hong2013}, these precoding weights are chosen from a finite cardinality codebook.

%The main challenge in this scenario is that the joint problem of clustering and precoding for multiple sources is an NP-hard problem. Being NP-hard implies that the algorithm is not scalable in the number of \UEs. Even though we separate these two problems, the sub-problem of  precoding still remains NP-hard, and therefore, approximation methods must be used to solve it in polynomial time. 

To the best of our knowledge, this is the first work that studies the network performance as a function of the precoding and density of randomly distributed \UEs~and proposes an \emph{efficient}  clustering algorithm for the VMIMO setup (as modelled in Section \ref{sec:system_model_vmimo}). Towards this goal, in Section \ref{sec:precoding_approximation}, we approximately solve the NP-hard problem of precoding optimization by a combination of  the search space reduction  and semidefinite relaxation (SDR). In Section \ref{sec:one_stue}, we study a special case with a single source and analytically derive an upper bound on the achievable spectral efficiency. This upper bound encompasses the stochastic geometry of the problem as well as the randomness in channel gains due to the log-normal shadowing. Then we propose a greedy algorithm with quadratic complexity in the number of \UEs~that achieves this bound, and hence, it is optimal in this case. %The key observation in this algorithm is that to form the VMIMO cluster, the ordered set of candidate relays are needed to be swept only once. 
We then leverage this knowledge in Section \ref{sec:clustering_alg_multi_source} to develop a  clustering   algorithm for multiple sources. Our proposed  algorithm is efficient in the following senses: firstly, it is run in polynomial time; secondly, it eliminates the need for the backhaul communication between APs since it is not required to control the leakage interference; thirdly, it can achieve significant performance gains (cf. Fig.  \ref{fig:SINR_imp_vs_base_SINR}, Fig. \ref{fig:SINR_imp_CDF}, and Table \ref{tbl:harmonic_mean_imp}). In Section \ref{sec:numerical_results}, numerical results are provided demonstrating, firstly, that the performance of our algorithm is very close to that of an exhaustive clustering, and secondly, that our algorithm can improve the tradeoff between the spectral efficiency and energy efficiency of the implementation.

\section{System Model and Problem Statement} \label{sec:system_model_vmimo}

\subsection{Notations and Definitions}

The terms \emph{rate} and \emph{spectral efficiency} are used interchangeably to indicate the average number of bits per second per hertz (bps/Hz) that can be conveyed through the communication channel. 
The channel is assumed to be a complex AWGN channel with noise distribution $\mathcal{CN}(0, \pn)$. By definition, let $C(x) =\log_2 (1 + x)$. The \emph{ceiling} of $x$ is shown by $\lceil x \rceil$. Calligraphic letters (e.g., $\mathcal{A}$) are used to represent sets. The backslash notation is used to represent the set subtraction, i.e., $\mathcal{A}\backslash \mathcal{B} \stackrel{\Delta}{=} \mathcal{A} - \mathcal{B}$. Vectors are shown by small bold faced letters (e.g., $\mathbf{h}$), and matrices are shown by capital bold faced letters (e.g., $\mathbf{H}$). The $i$-th element of a vector $\mathbf{h}$ is shown by $h_i$. The conjugate transpose of a matrix is shown by $\herm{(\cdot)}$. %The trace of a matrix is represented by $\text{tr}(\cdot)$. 
The operator $\abs{\cdot}$ is either the absolute value when its operand is a complex (or real) number or the set cardinality if it operates on sets. $I_N$ represents the unity matrix with dimension $N \times N$. 

\subsection{Distribution of \UEs~and Scheduling Model}

We assume that \UEs~are randomly distributed in a two dimensional field according to a Poisson point process with rate $\lambda$ \UEs~per unit area. The Poisson assumption has been generally accepted  as a proper model for the spatial distribution of \UEs~in wireless networks specially in the presence of a large population of users (cf.  \cite{Haenggi2009, Ilow1998, Buratti2011}).

Let $\xi$ be a realization of the spatial distribution of \UEs, i.e., there are $N_{\UE}(\xi)$ \UEs~on the field given this realization. We further assume that $N_{\UE} (\xi)$ remains stationary for a relatively long time. Therefore, we omit $\xi$ hearafter and let  \UEs~be indexed as $1, 2, \cdots, N_{\UE}$. Without loss of generality, we  assume there are $M$ APs, and the scheduler schedules one \UE~per AP; therefore, there are $M$ sources. Let these $M$ sources be indexed as $1, \cdots, M$ where $M \le N_{\UE}$. 

\begin{figure}[t]
\begin{center}	
 	\includegraphics[width = \figwidth]{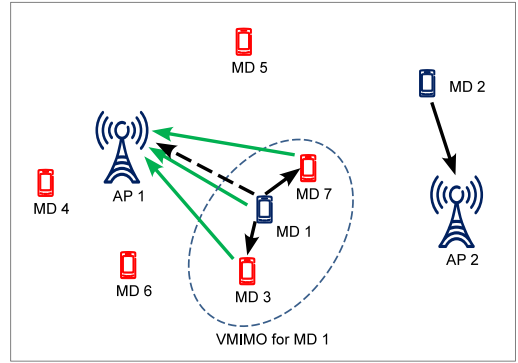}
	\caption{Illustration of the system model for $M=2$ destinations (APs) and 7 \UEs~($N_{\UE} = 7$).  In this case, \UEs~$1$ and $2$ are scheduled to transmit (sources), and \UEs~$3$ and $7$ are assisting the \UE~$1$ in the VMIMO configuration. Black lines represent transmission in the first phase, and green lines represent transmission in the second phase. 
	\label{fig:system_model}}
\end{center}
\end{figure}

\subsection{Structure of the VMIMO and Corresponding Rates}

It is assumed that each \UE~is equipped with one transmit antenna, and each AP is equipped with $N_{rx}$ receive antennas. Each source $s$ either  directly transmits to its respective destination without the assistance of other relays or adopts a two-phase transmission where it is assisted by the set of idle \UEs~(relays) $\mathcal{A}_s$. In the latter case, in phase one, the source broadcasts its message (codeword), and  all \UEs~$k$, $k\in\mathcal{A}_s$, decode the message of the source while the AP postpones the decoding. In phase two, the source $s$ and all the \UEs~$k$, $k\in\mathcal{A}_s$,  transmit the same message as that of the phase one with the precoding, and then, the AP decodes this message by augmenting received signals in phases one and two.
 In more precise words, if the source  $s$ is assisted by other \UEs, it transmits the codeword $\mathbf{x}_j^{(1)}$ in phase one and repeats the same codeword in phase two. On the other hand, if the source $s$ is not assisted by other \UEs, it transmits $\mathbf{x}_j^{(q)}$ in phase $q$ where $q=1,2$. 
%A random coding argument is used here, and therefore, it  can be  assumed that the elements of the codeword $\mathbf{x}_j$ are identically and independently distributed (i.i.d.) over time. Therefore, we consider a specific time index (counted from the beginning of the codeword) and omit the time index for the elements of these codewords. 
We further assume that the symbols in the a codeword are power normalized, i.e., $E[\abs{x_j}^2] = 1$.

Let $\mathbf{y}_d^{(1)}$ be the received signal in phase one at a specific time at the AP $d$. We  have 
%\begin{IEEEeqnarray}{rCl}
$\mathbf{y}_d^{(1)} = \sum_{j = 1} ^{M} \mathbf{h}_{jd} \sqrt{P_j} x_j^{(1)} + \mathbf{n}_d ^{(1)} $
%\end{IEEEeqnarray}
where $\mathbf{y}_d^{(1)}$ is an $N_{rx}\times 1$ vector, $\mathbf{h}_{jd}$ is an $N_{rx}\times 1$ vector whose elements are the channel gains between the \UE~$j$ and the AP $d$,  $x_j^{(1)}$ is  the power normalized symbol transmitted by the source $j$ in phase one,  $\mathbf{n}^{(q)}_d$ is the AWGN vector in phase $q$, and $P_j$ is the transmit power of the \UE~$j$. Also, let $\mathcal{L} = \set{1, \cdots, L}$ be the subset of sources that are assisted in the described VMIMO setup.  %and let $\mathcal{A}_k$ be the set of \UEs~that are assisting (relays) the source $k$ in the described VMIMO setup. 
 In the second phase, \UEs~that participate in VMIMO, precode their transmit signal with a unity complex weight factor. %This weight factor preserves the the transmit power and rotates the phase of the transmitted signal. 
 Therefore, the received signal at the AP $d$ in phase two can be written  as
%\begin{IEEEeqnarray}{rCl}
$\mathbf{y}_d^{(2)} = \sum_{j = 1} ^{L} \mathbf{H}_{jd} \mathbf{w}_j x_j^{(1)}  +\sum_{j = L+1} ^{M} \mathbf{h}_{jd} x_j^{(2)}  + \mathbf{n}_d ^{(2)}$
%\end{IEEEeqnarray}
where 
$
%\begin{IEEEeqnarray}{rCl}
 \mathbf{H}_{jd} = \begin{bmatrix} \mathbf{h}_{jd} \sqrt{P_{j}} & \mathbf{h}_{k_1d} \sqrt{P_{k_1}} & \cdots & \mathbf{h}_{k_{\abs{\mathcal{A}_j}}d} \sqrt{P_{k_{\abs{\mathcal{A}_j}}}} \end{bmatrix}$, $k_i \in \mathcal{A}_j
%\end{IEEEeqnarray}
$.
The column vector $\mathbf{w}_j$ contains $\abs{\mathcal{A}_j} + 1$ precoding weights $w_{jk}$ where 
%. As mentioned, these  weights are discrete phase rotators with unit magnitude, i.e.,  
$
%\begin{IEEEeqnarray*}{rCl}
 w_{jk} \in \set{1, w, \cdots, w^{\nw-1} }
%\end{IEEEeqnarray*}
$
and $w$ is the $\nw$-th principal root of unity. %, and $\nw$ is the precoding codebook size.
The augmented received signal at the AP $d$ after two phases is
\begin{IEEEeqnarray}{rCl}
\mathbf{y}_d &=& \sum_{j = 1}^{L}\begin{bmatrix}\mathbf{h}_{jd} \sqrt{P_j} \\ \mathbf{H}_{jd}\mathbf{w}_j\end{bmatrix} x_j^{(1)} + \sum_{j = L+1}^{M}\begin{bmatrix}\mathbf{h}_{jd}\sqrt{P_j} \\ 0 \end{bmatrix} x_j^{(1)} 
\nonumber \\&&
+ \sum_{j = L+1}^{M}\begin{bmatrix}0 \\\mathbf{h}_{jd}\sqrt{P_j} \end{bmatrix} x_j^{(2)} + \begin{bmatrix} \mathbf{n}_{d}^{(1)} \\ \mathbf{n}_{d}^{(2)}  \end{bmatrix}  .
\end{IEEEeqnarray}
% It is assumed  that the UL channel gains are known, i.e., the instantaneous channel state information (CSI) is available at the AP. Unlike most previous works (including\cite{Yu2012a, Beres2008, Zhou2013, Guo2010, Li2010, Wu2012}), the decoder at the destination AP $d$ utilizes an optimal linear MMSE decoder \cite[Chapter 8]{Tse2005} which exhibits a superior performance in interference dominated communications as compared to the maximum ratio combining or zero forcing. 

Employing a linear MMSE decoder \cite{Tse2005}, we first obtain the spectral efficiency of  sources that are not assisted by relays. For the source $s$, $s\in\{L+1, \cdots, M\}$, the capacity in phase $q = 1,2$ can be computed as
%\begin{IEEEeqnarray}{rCl}
$c_s^{(q)} = \C{P_s \herm{\mathbf{h}_{sd}} \mathbf{K}_{\mathbf{z}^{(q)}}^{-1}\mathbf{h}_{sd}}$
%\end{IEEEeqnarray}
where 
%$\mathbf{K}_{\mathbf{z}^{(q)}}$ is the covariance matrix of the overall noise and interference which is the white Gaussian noise $\mathbf{n}_d^{(q)}$ added with the interference induced by other transmitting \UEs.
\begin{comment}
This \emph{non-white noise} for phases one and two can be written as
\begin{IEEEeqnarray*}{rCl}
\mathbf{z}_k^{(1)} = \sum_{j=1, j\neq k}^{M}\mathbf{h}_{jd}{\color{red}\sqrt{p_j}}x_j^{(1)} + \mathbf{n}_d^{(1)},
\end{IEEEeqnarray*}
and 
\begin{IEEEeqnarray*}{rCl}
\mathbf{z}_k^{(2)} &=& \sum_{j = 1} ^{L} \mathbf{H}_{jd} \mathbf{w}_j x_j^{(1)}  +\sum_{j = L+1, j \neq k} ^{M} \mathbf{h}_{jd} {\color{red}\sqrt{p_j}} x_j^{(2)}  + \mathbf{n}_d ^{(2)}, 
\end{IEEEeqnarray*}
respectively. 
\end{comment}
%These noise covariance matrices can be computed as
\begin{IEEEeqnarray*}{rCl}
\mathbf{K}_{\mathbf{z}^{(1)}} &=& 
%E[\mathbf{z}_k^{(1)}\mathbf{z}_k^{{(1)}*}] % \nonumber \\ 
%=
\sum_{j=1, j\neq s}^{M}P_j \mathbf{h}_{jd} \herm{\mathbf{h}_{jd}} + \pn  I_{N_{rx}}, 
\\
%\end{IEEEeqnarray*}
%and
%\begin{IEEEeqnarray*}{rCl}
\mathbf{K}_{\mathbf{z}^{(2)}} %& =& E[\mathbf{z}_k^{(2)}\mathbf{z}_k^{(2)*}]  \nonumber\\  
& = & \sum_{j=1}^{L} \mathbf{H}_{jd} \mathbf{w}_{j} \herm{\mathbf{w}_{j}} \herm{\mathbf{H}_{jd}} + \sum_{j=L+1, j\neq s}^{M} P_j \mathbf{h}_{jd} \herm{\mathbf{h}_{jd}} 
%\nonumber \\ &&
+ \pn  I_{N_{rx}}.
\end{IEEEeqnarray*}
The overall rate of the source $s$ is the average of its rates over two phases, i.e., 
\begin{IEEEeqnarray}{rCl}\label{eqn:R_t_no_VMIMO}
r_s = \frac{1}{2}\left( c_s^{(1)} + c_s^{(2)}  \right), \quad s\in\set{L+1, \cdots, M}.
\end{IEEEeqnarray}

%Next, we obtain the rate of the sources that are assisted by other \UEs. 
%As mentioned, after the second phase, the AP decodes the message of the source by augmenting the signal from the source in the first phase with the signals from the source and the relays in the second phase. Therefore, 
The aggregate capacity between the source $s$ when assisted by the relays $\mathcal{A}_s$ in two phases and the destination $d$, can be written as 
\begin{IEEEeqnarray}{rCl}\label{eqn:aggregate_rate_assisted_ue}
c_s = \C{ \herm{\tilde{\mathbf{h}}_{sd}} \mathbf{K}_{\mathbf{z}_s}^{-1} \tilde{\mathbf{h}}_{sd} }
\end{IEEEeqnarray}
where 
$
%\begin{IEEEeqnarray}{rCl}
\herm{ \tilde{\mathbf{h}}_{sd}} = \begin{bmatrix} \herm{\mathbf{h}_{sd}} \sqrt{P_s} & \herm{\mathbf{w}_s} \herm{\mathbf{H}_{sd}} \end{bmatrix},
%\end{IEEEeqnarray}
$
and %$\mathbf{K}_{\mathbf{z}}$ is the covariance matrix of the non-white noise $\mathbf{z}_k$ where
%\begin{IEEEeqnarray}{rCl}
%\mathbf{z}_k &=& \sum_{j = 1, j \neq k}^L \tilde{\mathbf{h}}_{jd} x_j^{(1)} + \sum_{j= L+1}^M \begin{bmatrix} \mathbf{h}_{jd}{\color{red}\sqrt{p_j }} x_j^{(1)} \\ 0 \end{bmatrix} 
%\nonumber\\ &&
%+ \sum_{j= L+1}^M \begin{bmatrix} 
%0 \\ \mathbf{h}_{jd} {\color{red}\sqrt{p_j }} x_j^{(2)}
%\end{bmatrix} + \begin{bmatrix} \mathbf{n}_{d}^{(1)} \\ \mathbf{n}_{d}^{(2)}  \end{bmatrix}.
%\end{IEEEeqnarray}
\begin{IEEEeqnarray}{rCl}
\mathbf{K}_{\mathbf{z}_s} &=& \sum_{j = 1, j \neq s}^L \tilde{\mathbf{h}}_{jd} \herm{\tilde{\mathbf{h}}_{jd}} + \sum_{j= L+1}^M \begin{bmatrix} P_j \mathbf{h}_{jd} \herm{\mathbf{h}_{jd}} & 0 \\ 0 &  P_j  \mathbf{h}_{jd} \herm{\mathbf{h}_{jd}} \end{bmatrix} 
\nonumber\\ &&
 + \pn I_{2N_{rx}}
\end{IEEEeqnarray}
%Since it is assumed that all the relays (\UEs~in $\mathcal{A}_s$) decode the message of the source $s$ and forward it, the overall rate of the relay assisted source $s$ is the minimum of individual rates to each \UE~in set $\mathcal{A}_s$ and the aggregate rate $c_s$ computed in \eqref{eqn:aggregate_rate_assisted_ue}. In addition, there is a half factor since the assisted source $s$ transmits a repeated version of its message in two consecutive time slots.  
Therefore, the overall spectral efficiency of the relay assisted source $s$ can be written as
\begin{IEEEeqnarray}{rCl}\label{eqn:rate_vmimo_ue}
r_s = \frac{1}{2}\min\set{\set{c_s} \cup \set{r_{sl}: l \in \mathcal{A}_s}}, \quad s \in \mathcal{L} 
\end{IEEEeqnarray}
where $r_{sl}$ 
% = \C{ \abs{h_{kl}}^2 p_k / \pn}$ 
is the rate between the source $s$ and \UE~$l$ (computed similar to $c_s^{(1)}$), and $c_s$ is the rate in \eqref{eqn:aggregate_rate_assisted_ue}. The half factor is present since the assisted source $s$ transmits a repeated version of its message in two consecutive time slots. 
% Finally, it should be emphasized  that a source employs other \UEs~to form the VMIMO only if $R_s$ in \eqref{eqn:rate_vmimo_ue} is greater than its rate without the assistance from the VMIMO. %The factor $\frac{1}{2}$ in \eqref{eqn:rate_vmimo_ue} is due to the fact that the source consumes twice as much time to transmit the information. 

Fig. \ref{fig:system_model} illustrates the above described system model. In this figure, %\UEs~1 and 2 are the sources that are scheduled to transmit to their respective APs. The \UE~1 is assisted by idle \UEs~3 and 7 to form a VMIMO, i.e., 
$\mathcal{L} = \{1\}$, and $\mathcal{A}_1 = \{3,7\}$.

%In the recent releases of the LTE standard, a hybrid automatic repeat request (HARQ) coding  with the incremental redundancy (IR) is employed. In this case, the source would transmit a different redundancy version in phases one and two assuming that in phase one, transmission was received in error. In this work, however, we only consider the repetition-based redundancy. % since analyzing the coding schemes with the IR tend to be intractable. %Note that the more sophisticated HARQ coding with the IR usually exhibits a superior performance as compared with the repetition-based codes. Therefore, the results in this paper can serve as a lower bound or a bench mark for the general case with the IR.

\subsection{Statement of the Problem}

We seek to maximize the harmonic mean utility subject to disjoint set of assisting relays for each source and quantized unity precoding weights. This goal can be mathematically expressed as  
\begin{IEEEeqnarray}{rCl}\label{eqn:formal_problem}
\max & \quad& \frac{M}{\sum_{i=1}^M r_i^{-1}}\IEEEyessubnumber\label{eqn:formal_problem_objective}\\
\text{s.t.}&& w_{jk} \in \set{1, w, \cdots, w^{\nw-1}}, \; j\in\mathcal{L};\IEEEyessubnumber\\
 && \mathcal{A}_i \cap \mathcal{A}_j  = \emptyset, \quad i\neq j;  \IEEEyessubnumber\\
&& \cup_{j=1}^{M} \mathcal{A}_j \subseteq \set{M+1,\cdots,N_{\UE}}.\IEEEyessubnumber
\end{IEEEeqnarray}

%\begin{rem} In a multi-source network, there are generally four utility functions that can be maximized: weighted sum-rate utility, proportional fairness utility, harmonic mean utility, and min-rate utility \cite{Bjornson2013}. The last three utilities attain \emph{fairness} in the network with different tradeoffs between the aggregate throughput and fairness. In this paper, we seek to maximize the overall spectral efficiency of the network given that each source is allowed to transmit the same number  of information bits, and this objective leads to the harmonic mean utility. In more precise words, assume that  each source $i$ is allowed  to transmit $\kappa$ bits of  information conveyed in the time interval $T_i$ where $\kappa = T_i R_i$.  Therefore, the overall average spectral efficiency is given as the ratio of the total number of bits to the total time, i.e.,  $\kappa M/ \sum_i{T_i}$. Simplification of this ratio leads to the harmonic mean of individual rates in \eqref{eqn:formal_problem_objective}. \end{rem}

\section{An Approximation Method for the Precoding Problem} \label{sec:precoding_approximation}

%Precoding is preformed when users are clustered together as a VMIMO device to improve the performance by forming a \emph{beam} towards the desired receiver. This method is also know as the transmit beamforming. To achieve an optimal performance, the precoding problem and the clustering problem should be considered jointly. However, due to the exponential complexity of the joint problem, they are considered separately in this work.

To approach the NP-hard problem of precoding, we isolate the precoding search space for each VMIMO cluster (associated with one source) and show in the following that SDR  can be applied to obtain an approximate solution.

The SINR term in \eqref{eqn:aggregate_rate_assisted_ue} can be written as 
\begin{IEEEeqnarray}{rCl}\label{eqn:sinr_aggregate_rate_assisted}
% \IEEEeqnarraymulticol{3}{l}{
\herm{ \tilde{ \mathbf{h} }_{sd}} \mathbf{K}_{\mathbf{z}_s}^{-1}\tilde{ \mathbf{h} }_{sd} 
% } \nonumber\\ \quad
& = & 
\begin{bmatrix} \herm{\mathbf{h}_{sd}} \sqrt{P_s} & \herm{\mathbf{w}_s} \herm{\mathbf{H}_{sd}} \end{bmatrix} \mathbf{K}_{\mathbf{z}_s}^{-1}
\begin{bmatrix}  \mathbf{h}_{sd} \sqrt{P_s} \\ \mathbf{H}_{sd} \mathbf{w}_s \end{bmatrix}
\nonumber \\ 
%&=& 
%p_k 
%\begin{bmatrix} 1 & \herm{\mathbf{w}_k} \end{bmatrix} 
%\begin{bmatrix} \herm{\mathbf{h}_{kd}} & 0 \\ 0 & \herm{\mathbf{H}_{kd}} \end{bmatrix} \mathbf{K}_{\mathbf{z}}^{-1}
%\begin{bmatrix} \mathbf{h}_{kd} & 0 \\ 0 & \mathbf{H}_{kd} \end{bmatrix} 
%\begin{bmatrix} 1 \\ \mathbf{w}_k \end{bmatrix} \nonumber \\ 
&=& 
 \herm{\tilde{ \mathbf{w} }_s} \mathbf{Q}_s \tilde{\mathbf{w} }_s
\end{IEEEeqnarray}
where $\tilde{\mathbf{w} }_s = \herm{ \begin{bmatrix} 1 & \herm{\mathbf{w}_s} \end{bmatrix}}$, and 
\begin{IEEEeqnarray}{rCl}
\mathbf{Q}_s & = & 
{\textstyle 
\begin{bmatrix} \herm{\mathbf{h}_{sd}} \sqrt{P_s} & 0 \\ 0 & \herm{\mathbf{H}_{sd}} \end{bmatrix} \mathbf{K}_{\mathbf{z}_s}^{-1}
\begin{bmatrix} \mathbf{h}_{sd} \sqrt{P_s} & 0 \\ 0 & \mathbf{H}_{sd} \end{bmatrix}.
}
\end{IEEEeqnarray}
%Provably, $\mathbf{Q}_k$ is a Hermitian matrix. 
%As mentioned, we isolate the search space for precoding weights for each individual cluster; therefore,  
The maximization of the SINR in \eqref{eqn:sinr_aggregate_rate_assisted} is tantamount to the following maximization problem
\begin{IEEEeqnarray}{rCl} \label{eqn:precoding_fixed_A}
\max &\quad & \herm{\tilde{ \mathbf{w} }_s} \mathbf{Q}_s \tilde{ \mathbf{w}}_s \IEEEyessubnumber\\
\text{s.t.} && \tilde{w}_{si} \in \set {1, w, \cdots, w^{\nw - 1}}.  \IEEEyessubnumber 
\end{IEEEeqnarray}
The optimization problem expressed in \eqref{eqn:precoding_fixed_A} is called the \emph{discrete complex quadratic optimization problem} which belongs to the class of NP-hard problems \cite{So2007}. However, it can be approximately solved by using the SDR method. 
In this paper, we first use the CVX package \cite{cvx} to solve the relaxed version of \eqref{eqn:precoding_fixed_A}. %\eqref{eqn:sd_problem}. 
Next, we adopt a rank-one approximation \cite{Luo2010} in addition to the uniform quantization to find the precoding weights from the solution of the semidefinite relaxed problem.

Assuming a limited feedback, each \UE~participating in the VMIMO setup receives  the index of the precoding weight in the  precoding codebook, $\{1, w, \cdots, w^{\nw -1}\}$,  
% a \emph{precoding weight indicator} (PWI) from the serving AP 
 through some feedback mechanism. Since the size of the precoding codebook is $\nw$, this feedback requires  $\log_2(\nw)$ bits.

\section{Clustering for a Single Source} \label{sec:one_stue}

In this section we develop an optimal clustering algorithm when there is only one source ($M=1$) and one destination. This algorithm will be later used in Section \ref{sec:clustering_alg_multi_source} for the general case with multiple sources. 

To make computations tractable, we temporarily assume that $N_{rx} = 1$; therefore, the channel gain between users $l$ and $m$ can be expressed as $h_{lm} = \abs{h_{lm}}e^{j\theta_{lm}}$. We further assume that channel gains are  influenced by the path loss (PL) and log-normal shadowing\footnote{Note that the channel gain is usually composed of three major components: path loss, slow fading (shadowing) modelled as a log-normal random variable, and fast fading modelled as a Rayleigh random variable. However, the mean channel gain is only affected by the path loss and shadowing.}; therefore,  %the squared  magnitude of the channel gains can be expressed as
\begin{IEEEeqnarray}{rCl}\label{eqn:channel_gain_model}
\abs{h_{lm}}^2 = Gd_{lm}^{-\alpha} 10 ^{\sigma_{dB} V_{lm}/10 }
\end{IEEEeqnarray}
 where $G$ is a constant depending on the operating frequency and the antenna gains, $d_{lm}$ is the distance between users $l$ and $m$, $\alpha$ is the path loss exponent, $\sigma_{dB}$ is the shadowing dB-spread, and  $V_{lm}$ is a standard Gaussian random variable. 
 In addition, it is assumed that  \UEs~are power controlled and the received SNR at the AP $d$ is constant, i.e.,
\begin{IEEEeqnarray}{rCl}
\abs{h_{ld}}^2  \frac{P_l}{\pn} = \gamma, \quad l = 1, \cdots, N_{\UE}.
\end{IEEEeqnarray}

\subsection{An Upper Bound on the Average Spectral Efficiency}

In order to evaluate any algorithm, a performance bound on the achievable spectral efficiency is needed. This performance bound is given in the following theorem. 

\begin{thm}\label{thm:average_throughput}
For the system described in this section, there exists an upper bound on the average spectral efficiency (bps/Hz) of the source $s$. For  $\sigma_{dB} = 0$, this upper bound is given as  
\begin{IEEEeqnarray}{rCl}\label{eqn:thm_avg_thp}
\C{\gamma} + \int_{\C{\gamma}}^{\infty}{\sum_{k \ge k_{r}}  {\frac{\left[\lambda \Psi_s(r) \right]^k } {k!}   e^{-\lambda \Psi_s(r) }  \textnormal{d}r  }   }
\end{IEEEeqnarray}
where $\lambda$ is the average \UE~density (number of \UEs~per unit area), and $\Psi_s(r)$ is the \emph{$r$-achievablity area} and given by
\begin{IEEEeqnarray}{rCl}
\Psi_s(r) = \pi \left[\frac{GP_s/\pn}{2^{2r}-1}\right]^{{2}/{\alpha} }, 
\end{IEEEeqnarray}
and $k_r$ is the \emph{$r$-necessary number of relays} and given by
\begin{IEEEeqnarray}{rCl}\label{eqn:r_neccessary}
k_{r} = \left\lceil   \sqrt{\frac{2^{2r} - 1 - \gamma}{\gamma}} - 1     \right\rceil.
\end{IEEEeqnarray}
\end{thm}

\begin{proof}
See Appendix \ref{sec:proof_thm_avg_thp_single_source}. 
\end{proof}

There are a number of important implications associated with Theorem \ref{thm:average_throughput}. Firstly, the second term ($\int(\cdot)\textnormal{d}r$) in \eqref{eqn:thm_avg_thp} represents  the average improvement in the spectral efficiency achieved by the formation of the VMIMO as compared with the baseline spectral efficiency $\C{\gamma}$ without VMIMO. In other words, if the source  can not benefit from the assistance of other \UEs, the second term in \eqref{eqn:thm_avg_thp} will be zero. Moreover, when the mean channel gains are influenced  by PL, the $r$-achievablilty area and $r$-necessary number of relays can be used to simplify the clustering in a practical implementation when the geographical location of \UEs~is known. 
We state these results in the following two corollaries. The proofs follow directly from the proof of Theorem \ref{thm:average_throughput}, and we skip them. 

%That is, if the average spectral efficiency predicted by Theorem \ref{thm:average_throughput} is $\bar{R}_t$, the area in which the  source expects to find assisting \UEs~is given by $\Psi_s(\bar{R}_s)$ centred at the source. Therefore, a further design implication would be to search in this area would be enough to maximize the average spectral efficiency. This implication can limit the size of search space when the physical location of \UEs~are known. 
%In addition, {\color{blue}when the mean channel gains are affected by PL,} $r$-necessary number of relays indicates the required number of relays for an optimal performance. Consequently, these two quantities %($\Psi_s(\bar{R}_s)$ and $k_{\bar{R}_s}$) 
%can be used to simplify the clustering algorithm in a practical implementation. %Based on the above discussion, Corollaries \ref{cor:num_relays} and \ref{cor:achievablity_area} immediately follow.    

\begin{cor}\label{cor:num_relays} Let $R_s$ be the maximum average spectral efficiency expressed in \eqref{eqn:thm_avg_thp}, and the channel gains are affected by the PL. The number of relays necessary to achieve $R_s$ is given by the $R_s$-necessary number of relays, $k_{R_s}$. \end{cor}

\begin{cor}\label{cor:achievablity_area} Let $R_s$ be the maximum average spectral efficiency expressed in \eqref{eqn:thm_avg_thp}, and the  channel gains are affected by the PL. The expected area in which assisting \UEs~are located is bounded by a disk centred  at \UE~$s$ with radius  $\sqrt{\Psi_s(R_s) / \pi}$. \end{cor}
% {\color{red} Add some figures quantifying the corollaries. }

For the general case with the log-normal shadowing, i.e.,  $\sigma_{dB} > 0$, an asymptotic upper bound can be obtained. The derivation of this bound is rather involved; we leave the details of this derivation to Appendix \ref{sec:upper_bound_shadowing} and state the results in the following lemma. 

\begin{lem}\label{thm:average_throughput_lns}
For the system described in this section and some positive $\delta$, assume that relays are located within a disk with radius $d_{\max}$ ($d_{\max} > \delta > 0$) centred at the source $s$. Let the probabilities $\pi_k(r, \delta)$ be calculated according to \eqref{eqn:pi_k}. When $\delta \rightarrow 0$, the average spectral efficiency of the source $s$ can be upper bounded as  
\begin{IEEEeqnarray}{rCl}\label{eqn:thm_avg_thp_lns}
\C{\gamma} + \int_{\C{\gamma}}^{\infty}{\Big[\sum_{k \ge k_{r}}  \pi_{k}(r, \delta)  \Big]  \textnormal{d}r     }.
\end{IEEEeqnarray}  
\end{lem}

Note that the mean of the shadowing term in \eqref{eqn:channel_gain_model} is greater than one, i.e., $E[10 ^{\sigma_{dB}  V_{lm}/10 }] \ge 1$. In other words, it is expected that the log-normal shadowing increases the connectivity  of the network, consistent with the results shown in \cite{Muetze2008}. This fact results in the bound for the average spectral efficiency in Lemma 1 being greater than that in Theorem \ref{thm:average_throughput}. This difference is illustrated in Fig. \ref{fig:throughput_vs_density}, and as can be seen, the log-normal shadowing increases the average spectral efficiency. It is also illustrated that these bounds are tight for infinite number of precoding weights.

\subsection{An Optimal $\bigO{{N_{\UE}}^2}$ Algorithm} 

For a fixed realization of \UEs, let $s$ be the source, and let the set of available idle \UEs~(relays) be $\mathcal{I} = \set{1,\dots, N_{\UE}} \backslash \set{s}$. Denote as Algorithm \ref{alg:culstering_one_stue}, an algorithm that can perform optimal clustering when $\nw = \infty$. Under this condition,  an optimal precoding can be performed by multiplying the signal $x_l$ by the precoding weight $e^{-j\theta_{ld}}$, thus cancelling the phase rotation cased by the channel.% gain $h_{ld}$.   

There are four major steps in Algorithm \ref{alg:culstering_one_stue}:  In the first step, the candidate relays are discovered. Candidate relays are the ones whose rate to the source is greater than twice as much as that of the source to the AP. In the second step, these candidate relays are sorted in a descending ranking based on their link rate to the source. In the third step, these ranked candidate relays are added to the cluster one by one if with  each addition, the spectral efficiency of the source increases. In the last step, it is decided whether the VMIMO system is formed or the source transmits in a one-phase manner.

The precoding step in Algorithm \ref{alg:culstering_one_stue} for scalar channel gains can be performed in $\bigO{N_{\UE}}$ steps. Therefore, the overall complexity of the algorithm would be $\bigO{N_{\UE}^2}$.

% < < < < < < < < < < < < < < < < < < < < < < < < 
%\begin{comment}

\begin{algorithm}[h]
\caption{Clustering for the system described in Sec. \ref{sec:one_stue}}
\label{alg:culstering_one_stue}
\begin{algorithmic}%[1]
	\State $r_s \gets \C{\gamma}$;  \Comment{Link rate without VMIMO}
	\State $\mathcal{E}_s \gets \set{l: l\in\mathcal{I}, r_{sl} > 2 r_s}$ \Comment{Candidate relay discovery}
	\State $\mathcal{E}_s^{ \text{sorted} } \gets $ Sort $\mathcal{E}_s$ descending w.r.t. $r_{sj}$, $j \in \mathcal{E}_s$
	\State $\mathcal{A}_s \gets \emptyset$; $r_{\text{old}} \gets 0$
	\For{$j$  \textbf{in} $\mathcal{E}_s^{\text{sorted} }$ }
		\State $ {\mathcal{A}}_s \gets \mathcal{A}_s \cup \set{j}$ 
		\State	$r_{\min} \gets \min \set{r_{sl}: l \in {\mathcal{A}}_s  }$
		\State $p({\mathbf{w}})  \stackrel{\Delta}{=}  \norm{\mathbf{h}_{sd}} ^ 2 \frac{P_s}{\pn}  + \frac{1}{\pn} \norm{  \sum_{ i\in \set{s}\cup \mathcal{A}_s }\mathbf{h}_{id} w_{i} \sqrt{P_i} } ^ 2  $
		\State $\gamma_{\text{AP}} \gets \max _{\mathbf{w}} p({\mathbf{w}})$
		\Comment{Precoding} 
		  %\State $\gamma_{\text{AP}} \gets \max _{\mathbf{w}}  \norm{\mathbf{h}_{sd}} ^ 2 \frac{P_s}{\pn}  + \frac{1}{\pn} \norm{  \sum_{ i\in \set{s}\cup \mathcal{A}_s }\mathbf{h}_{id} w_{i} \sqrt{P_i} } ^ 2  
		\State $r_{\text{new}} \gets \frac{1}{2} \min \set {\C{ \gamma_{\text{AP}} } ,\; r_{\min} }$
		\If {$r_{ \text{ new}} > r_{\text{old}}$}
			\State $r_{\text{old}} \gets r_{\text{new}} $
		\Else
			\State $\mathcal{A}_s \gets {\mathcal{A}}_s \backslash \{j\} $
		\EndIf
	\EndFor
	\If {$r_{\text{old}} < r_s$}	
		\State $\mathcal{A}_s \gets \emptyset$  \Comment{No VMIMO}
	\EndIf 
\end{algorithmic}
\end{algorithm}

%\end{comment}
% > > > > > > > > > > > > > > > > > > > > > > > > > > > > 

\begin{prop}\label{prop:best_relay_set}
Algorithm \ref{alg:culstering_one_stue} selects the best subset of \UEs~to cluster with the source, and therefore, on average it  can achieve the bound given in Theorem \ref{thm:average_throughput} for $\nw = \infty$. 
\end{prop}

\begin{proof} 
See Appendix \ref{sec:proof_prop_best_relay_set}. 
\end{proof}

\begin{rem} 
Algorithm \ref{alg:culstering_one_stue} can be considered  to belong to the popular category of  \emph{greedy} algorithms\cite{Cormen2009}, i.e., it chooses the best relay at each step. In fact, all algorithms (including the ones in \cite{Michalopoulos2006, Jing2009, Yu2012, Nam2008}) that sweep through a  sorted list of items  to select the best set of items fall under the category of greedy algorithms. However, there are four major differences between these works and our work in terms of the selection algorithm: firstly, in our work an adaptive number of relays is selected whereas the number of relays is fixed in \cite{Yu2012}; secondly, by considering the effect of the direct transmission, our algorithm chooses no relay when the source can not benefit from relays in terms of achievable rates whereas  in \cite{Jing2009} and \cite{Nam2008}, the assumption is that source has to select at least one relay; thirdly, unlike \cite{Nam2008, Yu2012}, Algorithm \ref{alg:culstering_one_stue} is not required to sort \emph{all} the relays by establishing a necessary  condition for the candidate relays in the discovery step of the algorithm, thus reducing the complexity; lastly, unlike \cite{Michalopoulos2006} and \cite{Jing2009}, our greedy algorithm is optimal. 
\end{rem}

Fig. \ref{fig:throughput_vs_density} compares the average spectral efficiency (bps/Hz) vs. $\lambda$ (\UEs/m$^2$) for different received SNRs and different schemes as described in Table \ref{tbl:lable_fig_one_stue}. For this figure,  Algorithm \ref{alg:culstering_one_stue} is performed over 100 random trials of spatial distribution of \UEs. As can be seen, the spectral efficiency gained by Algorithm \ref{alg:culstering_one_stue} overlaps with the upper bound on the average spectral efficiency given in Theorem \ref{thm:average_throughput} and Lemma \ref{thm:average_throughput_lns} when exact phase matching is performed, i.e., $\nw = \infty$ for graphs labeled (ii) and (iv). This observation further corroborates the optimality of Algorithm \ref{alg:culstering_one_stue} stated in Proposition \ref{prop:best_relay_set}. The precoding for this simulation (in graphs (v)--(vi)) is performed as described in Section \ref{sec:precoding_approximation}. It can be speculated that this approximate precoding method  improves the performance as compared to the case without the precoding ($\nw = 1$). Furthermore, the improvement in the spectral efficiency due to the VMIMO is more significant for poor performing \UEs. For instance, for $\lambda = 0.01$ \UEs/m$^2$, the VMIMO algorithm can boost the spectral efficiency of \UEs~with $\gamma = -10$ dB more than $400\%$ as compared with the baseline spectral efficiency of $\C{\gamma} = 0.14$ bps/Hz  whereas this improvement is about $10\%$ for \UEs~with $\gamma = 10$ dB.

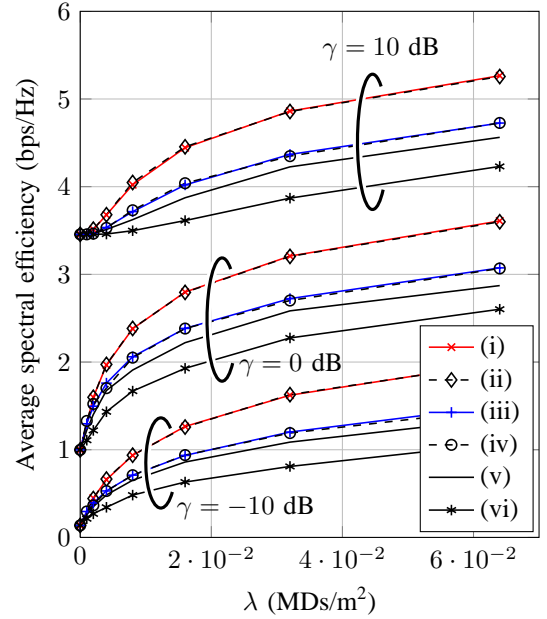
\begin{figure}[t]
\begin{center}	
 	%\includegraphics[width = \figwidth]{figs2/rate_vs_density_vs_SNR.eps}
	% double col: 
	\setlength\figureheight{7cm}  \setlength\figurewidth{6.1cm}
	% single col: 
	%\setlength\figureheight{10cm}  \setlength\figurewidth{10cm}
	{
		\input{fig_rate_vs_density_vs_SNR_dmax_25m.tikz}
	}
	
%	\draw[line width=3.5pt, draw=white] (1.2,1.45) arc  (50:320:0.2cm and 0.6cm) ;
%	\draw[very thick] (1.2,1.45) arc  (50:320:0.2cm and 0.6cm) ; 
%	\node [] at(1.9,0.2) {$\gamma = -10$ dB}; 
%	
%	\draw[line width=3.5pt, draw=white] (2,3.57) arc  (55:315:0.2cm and 0.82cm) ; 
%	\draw[very thick] (2,3.57) arc  (55:315:0.2cm and 0.82cm) ; 
%	\node [] at(2.8,2.3) {$\gamma = 0$ dB}; 
%	
%	\draw[line width=3.5pt, draw=white] (4,6) arc  (55:315:0.2cm and 0.9cm) ; 
%	\draw[very thick] (4,6) arc  (55:315:0.2cm and 0.9cm) ; 
%	\node [] at(4,6.5) {$\gamma = 10$ dB}; 
	
	\caption{Average spectral efficiency vs. density of the \UEs~for various received SNRs. Curves are grouped  for each received SNR $\gamma$, and the labels in the legend are explained in Table \ref{tbl:lable_fig_one_stue}. Note that $\lambda = 0$ is treated as a special case in which there is only one source and no other \UEs. $d_{\max} = 25m$ for (i) and (ii), and $\delta = 0.05$ for (i).
	 \label{fig:throughput_vs_density}}
\end{center}
\end{figure}

\begin{table}[t]
	\begin{center}
		\caption{Description  of the legends in Fig. \ref{fig:throughput_vs_density} } \label{tbl:lable_fig_one_stue}
		\begin{tabular}{c|c|c|c}
			Label & Method & $\sigma_{dB}$  & \nw \\ 
			\hline
			(i) & Lemma \ref{thm:average_throughput_lns} & $ 8$ & N/A \\
			(ii) & Algorithm \ref{alg:culstering_one_stue} &  $ 8$ & $\infty$\\
			(iii) & Theorem \ref{thm:average_throughput} & 0 & N/A\\
			(iv) & Algorithm \ref{alg:culstering_one_stue} & $0$& $\infty$\\
			(v) & Algorithm \ref{alg:culstering_one_stue} & $0$ &  $8$\\
			%(vi) & Algorithm \ref{alg:culstering_one_stue} & PL \& Rayleigh fading &  $\nw = 8$\\
			(vi) & Algorithm \ref{alg:culstering_one_stue} & $0$ & $1$ 
			%(viii) & Algorithm \ref{alg:culstering_one_stue} & PL \& Rayleigh fading & $\nw = 1$\\
		\end{tabular}
	\end{center}
\end{table}

\section{Clustering for Multiple Sources} \label{sec:clustering_alg_multi_source}

In this section, we extend the algorithm developed in Section \ref{sec:one_stue} to the general case with multiple sources and develop Algorithm \ref{alg:culstering_multi_stue}. Algorithm \ref{alg:culstering_one_stue} is used in Algorithm \ref{alg:culstering_multi_stue} as a sub-algorithm with the following  modifications: firstly, $\C{\cdot}$ is obtained by the general formula given in \eqref{eqn:R_t_no_VMIMO} and \eqref{eqn:aggregate_rate_assisted_ue}; secondly, the precoding is performed according to Section \ref{sec:precoding_approximation}. 

Unlike Algorithm \ref{alg:culstering_one_stue}, we were unable to prove an optimality bound for Algorithm \ref{alg:culstering_multi_stue}. However, through extensive simulations, we later illustrate %(see Table \ref{tbl:harmonic_mean_imp} and Fig. \ref{fig:SINR_imp_vs_base_SINR}) 
that Algorithm \ref{alg:culstering_multi_stue} exhibits a near optimal performance when compared to that of an exhaustive clustering algorithm as far as the harmonic mean utility metric is concerned. 

\begin{algorithm}[h]
\caption{Clustering algorithm for multiple sources}
\label{alg:culstering_multi_stue}
\begin{algorithmic}%[1]
	\State $\mathcal{M} \gets \set{1, \cdots, M}$ % \Comment{Set of sources}
	\State $\mathcal{I} \gets \set{M+1, \cdots, N_{\UE}} $ \Comment{Available idle \UEs}
	\State $\mathcal{M}^{\text{sorted}} \gets$ Sort $\mathcal{M}$ ascending w.r.t. $c_s$, $s \in \mathcal{M}$
	\For {$s$ \textbf{in} $\mathcal{M}^{\text{sorted}}$}
		\State Perform Algorithm \ref{alg:culstering_one_stue} for \UE~$s$ and idle \UE~set $\mathcal{I}$
		\State $\mathcal{I} \gets \mathcal{I} \backslash \mathcal{A}_s $
	\EndFor
\end{algorithmic}
\end{algorithm}

To motivate why a simple source sorting results in a near optimal performance, it should be noted that %our objective is to attain fairness by considering the harmonic mean utility function. 
one important property of the harmonic mean is that it is limited by the smallest term\footnote{
To clarify, consider the following inequality that holds for positive numbers $r_1$, $r_2$, $\cdots$, $r_M$: 
$
{M}/({\sum_{i=1}^M r_i^{-1}}) \le M \min_{1 \le i \le M}{r_i}
$.
}. This property implies that an increase in the smallest term can potentially  lead to a significant increase in the harmonic mean. Therefore, in Algorithm \ref{alg:culstering_multi_stue}, \emph{underprivileged} sources are served first so as to enable them to benefit from the most available resources.

%One major issue when performing clustering for multiple sources is that as the number of assisting \UEs~increases the  interference to non-intended APs also increases. 
%As shown later in the numerical results (Fig. \ref{fig:SINR_imp_vs_base_SINR}), the leakage interference tends to slightly diminish  the spectral efficiency of \emph{high rate} sources which can be alleviated by controlling this interference. Controlling this leakage interference, nonetheless, imposes a high traffic load on the backhaul network since the interference information is needed to be constantly exchanged between APs as the algorithm evolves. 

From an algorithmic perspective, we can distinguish Algorithm \ref{alg:culstering_multi_stue} from previous works on DF relay selection for multiple sources regarding three aspects. Firstly, our algorithm selects an adaptive number of relays for each source while in \cite{Beres2008} and \cite{Zhou2013}, one relay is selected for each source. As we argued, enforcing the source to select relays is not always optimal in terms of the spectral efficiency. %For instance, for high SINR sources, when the density of \UEs~is low, we see (cf. Fig. \ref{fig:num_relays_vs_base_SINR}) that no relay is selected. 
Secondly, even though Algorithm \ref{alg:culstering_multi_stue} is oblivious to the leakage interference, the users are assumed to share the same channel. To isolate the relay selection for each user, however, in \cite{Beres2008} and \cite{Yu2012a}, it is assumed that sources utilize orthogonal channels. 
%; hence not applicable to the case where users share the same channel.  
%By this assumption, for instance in \cite{Yu2012a}, the choice of relays by a first source does not impact the weights for a second source (on a different channel).
Thirdly, in \cite{Yu2012a, Zhou2013, Li2010}, conflicts in selection are resolved by deploying some variation of the \emph{message passing} procedure \cite{Frey2007} (also known as auction rounds). For instance, \cite{Yu2012a} and \cite{Zhou2013} require up to $M$ and $M-1$ rounds of iteration, respectively. This number can be even larger in \cite{Li2010} depending on system parameters. 
%Although  iterative message passing algorithms are widly used when the clustering of objects is considered \cite{Frey2007}, they may not be suitable for wireless communication scenarios. 
In other words, the number of iteration rounds scale with the number of sources, and therefore, these algorithms may not be implementable in a large scale scenario where the channels are dynamically and rapidly changing.
 This problem is addressed in our work by prioritizing sources and requiring only one round of iteration.  
Consequently, we show that a \emph{simple} and \emph{agnostic} algorithm can perform nearly optimal when the harmonic mean utility metric is considered, and it may not be necessary to resort to a complex and time taking solution.

\section{Numerical Results} \label{sec:numerical_results}

\subsection{Simulation Setup and Parameters}
In our simulations, we use the 3GPP channel model for the indoor environment \cite[Sec. A.2]{3gpp}. According to this model the PL (in dB) for $2$ GHz carrier frequency is given by 
%\begin{IEEEeqnarray*}{rCl} 
$103.4 + 24.2 \log(d)$ 
%\end{IEEEeqnarray*}
where $d$ is the distance in km. The slow fading is modelled by the log-normal shadowing with dB-spread $\sigma_{dB} = 8$. Antennas are assumed to be omnidirectional with $0$ dB antenna gain. The fast fading is assumed to be Rayleigh fading modelled as a Gaussian complex random variable with variance $1/2$ per real dimension.  %The variations in channel gains are assumed to be slow enough such that they can be assumed to remain constant over two consecutive transmission phases. 
The noise power is $\pn = -101 $ dBm for $20$ MHz bandwidth. The maximum transmit power of \UEs~is limited to $P_{\max} = 20$ dBm. \UEs~are assumed to be power controlled, and the power control algorithm is assumed to average out the effect of Rayleigh fading, i.e., its decisions are only based on the PL and shadowing. As a result of this power control, each \UE's power is adjusted such that the received power at the serving AP is $-80$ dBm excluding the interference. The number of receive antennas at the AP is $N_{rx} = 4$. The field size is 100m$\times$100m. We assume that there are five APs, and therefore, at each instance, there are five sources. Furthermore, to simulate the environments with interference, we assume there is an aggressor network with average density of $10^{-3}$ sources per square meter with a similar power control algorithm. The clustering algorithm has no control over this aggressor network. 

For the sake of a better visualization in our simulation results, we use the effective SINR (in dB) instead of the spectral efficiency. In other words, if $r_k$ is the spectral efficiency of the source $k$, the corresponding  effective SINR, SINR$_{\text{eff}}$, is calculated as 
\begin{IEEEeqnarray}{rCl}\label{eqn:sinr_eff}
\text{SINR}_{\text{eff}} = 10 \log_{10}\left( 2^{r_k} - 1 \right).
\end{IEEEeqnarray}
For each $\lambda$ or $\nw$, the reported simulation results are  averaged over 1000 trials. These trials encompass  the random spatial distribution of \UEs~and APs, random log-normal shadowing, random Rayleigh fading, and random scheduling.

\subsection{Discussion}

The performance of Algorithm \ref{alg:culstering_multi_stue} is illustrated in Figs. \ref{fig:SINR_imp_vs_base_SINR}, \ref{fig:num_relays_vs_base_SINR}, \ref{fig:SINR_imp_CDF}, Table \ref{tbl:harmonic_mean_imp}, and Table \ref{tbl:epb_increase}. In these illustrations, $\lambda$ represents the density of the network (\UEs/m$^2$), and $\nw$ represents the number of precoding choices (there is no precoding when $\nw = 1$). For low densities of \UEs, these performances are compared to that of an exhaustive clustering. However, since the run time of the exhaustive clustering grows exponentially with the density of \UEs,  for  high density networks, it is not computationally feasible to run the exhaustive algorithm. %Moreover, the precoding in Algorithm \ref{alg:culstering_multi_stue} is performed according to the approximation method  in Section \ref{sec:precoding_approximation}. 

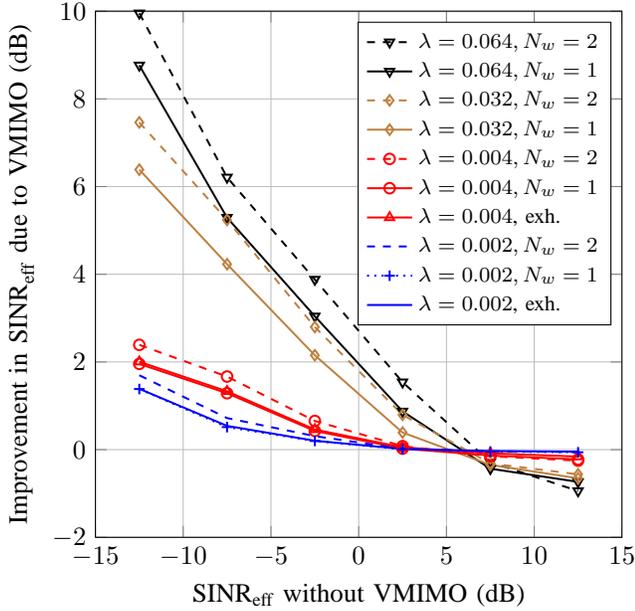
\begin{figure}[t]
\begin{center}	
 	% \includegraphics[width = \figwidth]{figs2/SINR_imp_vs_baseline_SINR.eps}
	% double col: 
	\setlength\figureheight{7cm} \setlength\figurewidth{7cm}
	% single col: 
	%\setlength\figureheight{10cm} \setlength\figurewidth{15cm}
	\input{fig_SINR_imp_vs_baseline_SINR.tikz}
	\caption{The improvement of the SINR$_{\text{eff}}$ after the formation of VMIMO using Algorithm \ref{alg:culstering_multi_stue} vs. the baseline SINR$_{\text{eff}}$ of sources without VMIMO. $\lambda$ represents the \UE~density (\UEs/m$^2$), and $\nw$ represents the number of precoding choices (there is no precoding for $\nw = 1$).
	 \label{fig:SINR_imp_vs_base_SINR}}
\end{center}
\end{figure}

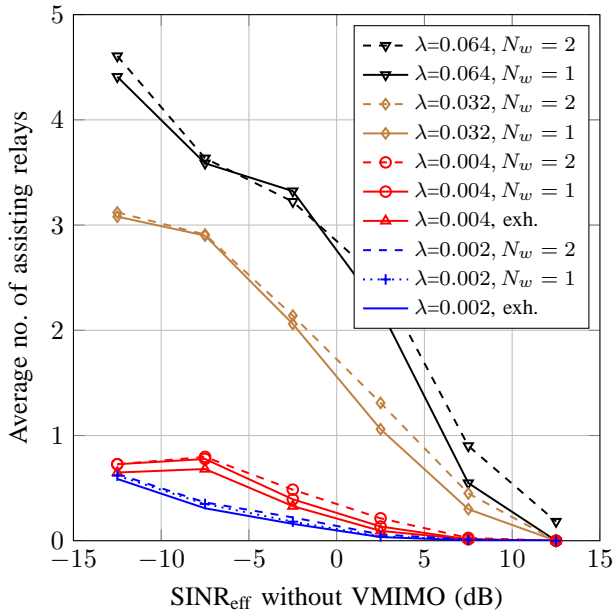
\begin{figure}[t]
\begin{center}	
 	%\includegraphics[width = \figwidth]{figs2/num_relays_vs_baseline_SINR.eps}
	% double col: 
	\setlength\figureheight{7cm} \setlength\figurewidth{7cm}
	% single col: 
	%\setlength\figureheight{10cm} \setlength\figurewidth{15cm}
	\input{fig_num_relays_vs_baseline_SINR.tikz}
	\caption{Average number of assisting \UEs~(relays) using  Algorithm \ref{alg:culstering_multi_stue} vs. the baseline SINR$_{\text{eff}}$ of sources without VMIMO. $\lambda$ represents the \UE~density (\UEs/m$^2$), and \nw represents the number of precoding choices (there is no precoding for $\nw = 1$). 
	 \label{fig:num_relays_vs_base_SINR}}
\end{center}
\end{figure}

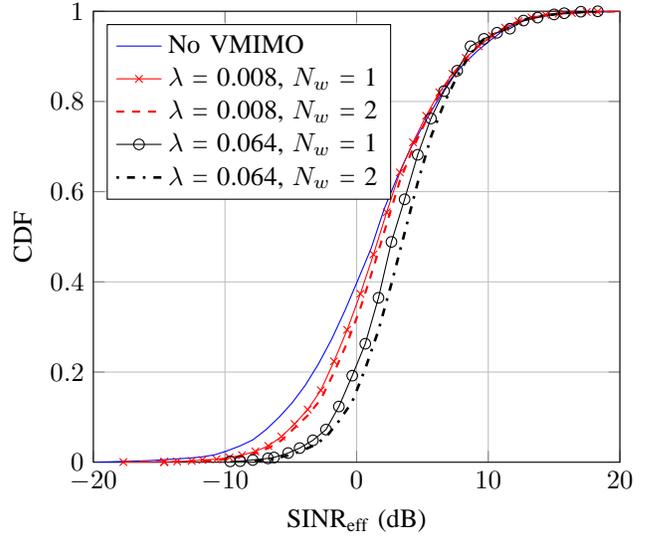
\begin{figure}[t]
\begin{center}	
 	%\includegraphics[width = \figwidth]{figs/SINR_imp_CDF.eps}
	% double col: 
	\setlength\figureheight{6cm} \setlength\figurewidth{7cm}
	% single col: 
	%\setlength\figureheight{10cm} \setlength\figurewidth{12cm}
	{
		\input{fig_SINR_imp_CDF.tikz}
		}
	\caption{Cumulative distribution function of the SINR$_{\text{eff}}$ for different densities of \UEs~and precoding resolutions.   $\lambda$ represents the \UE~density (\UEs/m$^2$), and $\nw$ represents the number of precoding choices (there is no precoding for $\nw = 1$). \label{fig:SINR_imp_CDF}}
\end{center}
\end{figure}

% Table \ref{tbl:harmonic_mean_imp} summarizes the percentage of improvement in the harmonic mean utility for different \UE~densities and the coverage threshold SINR$_{\text{eff}} = -10$ dB. The coverage threshold represents the coverage of the APs, i.e., only \UEs~with SINR$_{\text{eff}}$ more than $-10$ dB are granted access to the network. In these simulations,  the overall harmonic mean of the spectral efficiency without VMIMO is $0.83$ (bps/Hz). As can be seen, Algorithm \ref{alg:culstering_multi_stue} performs near optimal as compared to the exhaustive clustering at least for low density networks. % Also, these data show that the  precoding can significantly increase the spectral efficiency. %For instance, for $\lambda = 64\times 10^{-3}$ \UEs/m$^2$, precoding can achieve $15\%$ more improvement as compared to that  without the precoding. Another significant observation is that almost all of the gain due to the precoding can be achieved when there is only two choices of precoding weights ($\nw = 2$). In other words, 
%Also, these data show that only one bit of feedback for the PWI is \emph{sufficient} to provide the majority of the gain. 

Fig. \ref{fig:SINR_imp_vs_base_SINR} shows the improvement in the SINR$_{\text{eff}}$ of the sources after the formation of the VMIMO system vs. the baseline SINR$_\text{eff}$ of sources without the VMIMO. %The improvement is depicted for various \UE~densities when Algorithm \ref{alg:culstering_multi_stue} is performed.   
As can be seen, the spectral efficiency of \emph{poor} performing sources is significantly increased whereas the spectral efficiency of the \emph{strong} sources with high SINR$_{\text{eff}}$ is slightly degraded. For instance, when there is one \UE~in every 15 m$^2$ on average ($\lambda = 0.064$), sources with SINR$_{\text{eff}} = -10$ dB can gain more than $7$ dB improvement in their SINR$_{\text{eff}}$ while sources with SINR$_{\text{eff}} = 10$ dB  suffer less than $1$ dB. This compromise, nevertheless, leads to an overall  improvement (cf. Table \ref{tbl:harmonic_mean_imp}). %This figure also demonstrates that Algorithm \ref{alg:culstering_multi_stue} can perform very close to optimal when the improvement is compared to that of an exhaustive clustering. In the figure, the performance curves for exhaustive clustering and Algorithm \ref{alg:culstering_multi_stue} are overlapped, and it might be difficult to distinguish them.

Fig. \ref{fig:num_relays_vs_base_SINR} illustrates the average number of assisting \UEs~(relays) due to the VMIMO vs. the baseline SINR$_{\text{eff}}$ of sources without the VMIMO. %These numbers are depicted for various \UE~densities when Algorithm \ref{alg:culstering_multi_stue} is performed. 
As can be seen, the number of assisting \UEs~tend to increase when the precoding is performed because  the aggregate uplink rate in \eqref{eqn:aggregate_rate_assisted_ue} is increased due to this precoding. Hence, considering \eqref{eqn:rate_vmimo_ue}, the number of assisting relays can be increased.    %As the number of assisting  relays increases, the performance reaches a point of diminishing return because of randomness in the angle of arrivals of signals from \UEs~in the VMIMO system to the respective AP. The precoding alleviates this phenomenon by forming a beam towards the respective AP and optimizing the angle of arrivals in a constructive manner.

Fig. \ref{fig:SINR_imp_CDF} demonstrates the CDF of the SINR$_{\text{eff}}$ for different densities of \UEs~and precoding resolutions. As can be seen, the shift of the SINR$_{\text{eff}}$ towards higher values on the right, increases as the density of \UEs~increases. Also, the results clearly illustrate gains of up to 7 dB for the worst 5-10\% of users and a gain of 3 dB on average.

% Table \ref{tbl:epb_increase} shows the increase in the average consumed energy per information bit $\overline{E}_b$ due to the formation of the VMIMO. In our simulation, the  $\overline{E}_b$ for the network before implementing the VMIMO is equal to $69.5 \mu$J/b. As can be seen, for higher densities of the \UEs, more energy is consumed on average to convey each bit of information. This increase is due to higher number of assisting \UEs~for higher densities (cf. Fig. \ref{fig:num_relays_vs_base_SINR}). %The preceding, nonetheless, can considerably alleviate this energy consumption. 

The effect of the precoding codebook size ($\nw$) is illustrated in Tables \ref{tbl:harmonic_mean_imp} and \ref{tbl:epb_increase}. In our simulation, before implementing the VMIMO,  the overall harmonic mean of the spectral efficiency is $0.83$ (bps/Hz), and the $\overline{E}_b$ is $69.5 \mu$J/b. As can be seen, one bit of feedback is enough to provide the majority of the gains. Also, juxtaposing Table \ref{tbl:harmonic_mean_imp} and Table \ref{tbl:epb_increase} reveals compelling  results: firstly, by applying Algorithm \ref{alg:culstering_multi_stue}, the spectral efficiency and  energy efficiency are both improved for low densities of \UEs; %For instance, when $\lambda = 2\times10^{-3}$ \UEs/m$^2$ and $\nw = 2$ (one bit of feedback for the PWI), $3.6\%$ gain is achieved for the harmonic mean of the spectral efficiency, and the the consumed energy per bit is decreased by $6.4\%$ on average. 
secondly, for high densities of \UE s,  Algorithm \ref{alg:culstering_multi_stue} provides a tradeoff, i.e., the harmonic mean of the spectral efficiency is improved at the cost of higher energy per bit. However, the more  accurate the precoding is, the higher the gain in the harmonic mean of the spectral efficiency is, and the lower the energy price is.
%secondly, for higher densities of \UEs, the precoding can significantly increase the spectral efficiency while reducing the average energy consumption per information bit. %As a case in point, for $\lambda = 64\times 10^{-3}$ \UEs/m$^2$, with one bit of feedback for precoding weight index ($\nw = 2$), an increase of $57\%$ in the spectral efficiency can be achieved ($14\%$ more than the case without the precoding) while the energy inefficiency is increased only $68.6\%$ which is $12\%$ less than the case without the precoding. % As a general rule, it can be concluded that the improvement in the spectral efficiency due to the VMIMO setup is at the cost of more consumed energy per bit. This energy consumption, however, can be controlled by limiting the number of assisting \UEs. Nevertheless, investigating the trade-off between the energy efficiency and spectral efficiency is outside the scope of this work.

\section{Concluding Remarks}

In this paper, we developed an efficient clustering algorithm for the device-to-device assisted VMIMO systems with limited feedback, and investigated the effect of the approximate precoding (transmit beamforming) and the \UE~density on the performance of this VMIMO system. As observed in the numerical simulations, the VMIMO system can significantly boost the performance of \emph{weak} sources while only slightly degrading that of \emph{strong} ones, thus leading to a considerable overall performance increment. %For instance, for medium densities ($\lambda = 0.01$), the spectral efficiency of \UEs~with $\gamma = -10$ dB can be increased more than $400\%$ (Fig. \ref{fig:throughput_vs_density}), or for higher densities ($\lambda = 0.06$), \UEs~with $-10$ dB SINR$_{\text{eff}}$ can gain more than $7$ dB improvement  from the formation of the VMIMO system (Fig. \ref{fig:SINR_imp_vs_base_SINR}).  
These observations would suggest an approach of focusing the user clustering on weak users and not necessarily on strong users. In addition, it was shown that a single bit of feedback for the precoding weight  is sufficient to provide the majority of the gain (Tables \ref{tbl:harmonic_mean_imp} and \ref{tbl:epb_increase}). Nonetheless, only the repetition-based cooperative diversity gain has been exploited in this work; utilizing space-time-coded cooperative diversity or the multiplexing gain and studying the tradeoff between the diversity and multiplexing gains is a promising area  for future investigations.

\begin{table}[t]
	\begin{center}
		\caption{Improvement in the harmonic mean of the spectral efficiency for 
		the coverage threshold SINR$_{\text{{eff}} } = -10$ $d$B.} 
		\label{tbl:harmonic_mean_imp} 
		\begin{tabular}{c | c  c  c }
			\multirow{2}{*}{$\lambda$}& 
			%\multirow{2}{*}{\begin{tabular}{c}Exhaustive \\$\nw = 1$\end{tabular} } & 
			\multicolumn{3}{c}{Algorithm \ref{alg:culstering_multi_stue}}\\
			\cline{2-4}
			&   $\nw = 1$ & $\nw = 2$ &  $\nw = 64$ \\
			\hline
			%\hline
			$0.002$  & $2\%$  & $4\%$  & $4\%$ \\
			%\hline
			%$0.004$  & 6.2\% & 9.7\% & 9.7\% \\
			%\hline 
			$0.008$  & $13\%$ & $18\%$  & $18\%$ \\
			%\hline
			%$0.016$  & 16\% & 22\%  & 22\% \\
			%\hline
			%$0.032$  & 31\% &  41\%   & 42\% \\
			%\hline
			$0.064$ & $43\%$ &  $57\%$   & $58\%$ \\
		\end{tabular}
	\end{center}
\end{table}

\begin{table}[t]
	\begin{center}
		\caption{Increase in the average consumed energy per bit $\overline{E}_b $ (J/$b$) due to the VMIMO setup. 
		} 
		\label{tbl:epb_increase} 
		\begin{tabular}{c | c c c }
			\multirow{2}{*}{$\lambda $ }& 
			%\multirow{2}{*}{\begin{tabular}{c}Exhaustive \\ $\nw = 1$ \end{tabular} } & 
			\multicolumn{3}{c}{Algorithm \ref{alg:culstering_multi_stue}}\\
			\cline{2-4}
			&   $\nw = 1$ & $\nw = 2$ & $\nw = 64$ \\
			\hline
			%\hline
			$0.002$ & $0.2\%$  & $-6\%$ & $-6\%$ \\
			%\hline
			%$0.004$ & 4.5\% & 0\% & 0\% \\
			%\hline 
			$0.008$ & $11\%$ & $3\%$ & $3\%$\\
			%\hline
			%$0.016$  & 21.3\% & 7.6\%   & 7.6\%\\
			%\hline
			%$0.032$ & 39.5\% &  29.5\%  & 29.3\% \\
			%\hline
			$0.064$  & $81\%$ &  $69\%$  & $68\%$ \\
		\end{tabular}
	\end{center}
\end{table}

\section{Acknowledgement}
The first author would like to thank Mr. Adrien Comeau of Ericsson Canada for his constructive feedbacks on this subject.

\appendices

\section{Proof of Theorem \ref{thm:average_throughput}} \label{sec:proof_thm_avg_thp_single_source}
%\begin{proof}
In what follows, denote by $\mathcal{E}_s(r)$ (or $\mathcal{E}_s$ for brevity) the set of all \emph{eligible} relays to participate as a DF relay for the source $s$ when its spectral efficiency is set to $r$. In other words, $\mathcal{E}_s$ includes relays whose link budget to the source  $s$ allow an achievable rate  greater than $2r$. The condition for the relay $l$ to be in the set $\mathcal{E}_s$ can be written as 
\begin{IEEEeqnarray}{rCl}
l \in \mathcal{E}_s(r) &\leftrightarrow& \C{ \abs{h_{sl}}^2\frac{P_s}{\pn} } \ge 2r. 
\label{eqn:eligibility_condition} 
\end{IEEEeqnarray}

To derive the ASE, the probability of the event $\set{R_{s} > r}$ for a given $r$ and a realization of the spatial relay distribution and shadowing needs to be obtained where $R_s$ is the overall spectral efficiency of the source $s$. 
%Note that $R_{s}$ is a random variable since both the spatial distribution of relays and the shadowing coefficients are random variables. 
To achieve this goal two different cases must be considered: 
\begin{itemize}
\item[1.] $r \le \C{\gamma}$: In this case, the source  transmits its message in one phase transmission without recruiting relays. Therefore, in this case, $\Pr\set{R_s > r} = \Pr\set{c_{sd} > r} = 1$.  \\
\item[2.] $r > \C{\gamma}$: In this case, a non-zero probability of $\set{c_{sd} > r}$ is feasible only if a spectral efficiency increment is achieved by employing  relays in the two-phase transmission. 
\end{itemize}
Considering the second case hereafter, by applying the law of total probability, it can be written  
\begin{IEEEeqnarray*}{rCl}
\Pr\set{R_{s} > r} = \sum_{k\ge 1}\Pr\set{R_{s} > r \big| |{\mathcal{E}_s}| = k}  \Pr\set{|{\mathcal{E}_s}| = k}.
\end{IEEEeqnarray*}
Considering \eqref{eqn:rate_vmimo_ue}, we can write 
\begin{IEEEeqnarray}{rCl}
\Pr\set{R_s > r \big| |{\mathcal{E}_s(r)}| = k} & =&   \Pr\set{c_s > 2r \big| |{\mathcal{E}_s}(r)| = k} \quad
\end{IEEEeqnarray}
where 
\begin{IEEEeqnarray}{rCl}\label{eqn:c_t_proof}
c_s &=& \C{ \abs{h_{sd}}^2 \frac{P_s}{\pn}+ \abs{  {\textstyle \sum_{i \in \set{s} \cup\mathcal{E}_s }{ h_{id}w_{i} \sqrt{ P_i } }  } }^2 / \pn}\nonumber\\
& = & \C{ \gamma  + \gamma \abs{ {\textstyle \sum_{i\in\set{s}\cup \mathcal{E}_s }  }  e^{j\hat{\theta}_i} }^2   }. 
\end{IEEEeqnarray}
$\hat{\theta}_i = \theta_i + \theta_{w_i}$ is the phase of the channel gain of \UE~$i$ to the AP when its phase is shifted by the preceding weight $w_i$. It can be seen form \eqref{eqn:c_t_proof} that $c_s$ is maximized when $\abs{{\textstyle \sum_{i\in\set{s} \cup \mathcal{E}_s }  }  e^{j\hat{\theta}_i} }$ is maximized. The maximum value of the latter quantity is $k+1$ when all $\hat{\theta}_i$'s are equal for $i \in \{s\}\cup \mathcal{E}_s $. Given this assumption, we can write
\begin{IEEEeqnarray}{rCl}
\IEEEeqnarraymulticol{3}{l}{
\Pr\set{c_s > 2r \big|  |\mathcal{E}_s(r)| = k} 
}\nonumber \\ \quad 
&\le&  \Pr\set{\C{ \gamma \left[1 + (k+1)^2\right] }> 2r  \big|  |\mathcal{E}_s(r)| = k}\nonumber\\
&=& \Pr\set{ k \ge k_r \big|  |\mathcal{E}_s(r)| = k } \label{eqn:pr_k_greater_k_r}
\end{IEEEeqnarray}
where $k_r$ can be easily computed and is given in the Theorem \ref{thm:average_throughput}. Moreover, the event $\set{|\mathcal{E}_s(r) | = k}$ can be translated as the event such that $k$ \UEs~(indexed by $l$, $l\in\set{1,\cdots,N_{\UE}}\backslash \set{s}$) are spatially distributed such that
\begin{IEEEeqnarray}{rCl}\label{eqn:c_s_t_r}
\C { \frac{G}{{d_{sl}} ^ {\alpha}} 10^{\sigma_{dB} V_l/10} \frac{P_s}{\pn}  } \ge 2r \label{eqn:c_greater_r}
\end{IEEEeqnarray}
where $V_l$ is a normal Gaussian random variable. 
The inequality in \eqref{eqn:c_greater_r} can be rearranged as 
\begin{IEEEeqnarray}{rCl}
d_{sl} e^{-\sigma V_l} \le   d_r
\label{eqn:threshold_distance}
\end{IEEEeqnarray}
where 
$%\begin{IEEEeqnarray*}{rCl}
\sigma = 0.1 \ln(10) \sigma_{dB} / \alpha
$%\end{IEEEeqnarray*}
, and
\begin{IEEEeqnarray*}{rCl}
d_r =   \left[\frac{G P_s / \pn}{2^{2r} - 1} \right] ^{1/\alpha}.
\end{IEEEeqnarray*}
Therefore, the eligibility condition in \eqref{eqn:eligibility_condition} can be simplified as 
\begin{IEEEeqnarray}{rCl}
l \in \mathcal{E}_s(r) &\leftrightarrow &  d_{sl} e^{-\sigma V_l} \le  d_r.
\label{eqn:eligibility_condition_2}
\end{IEEEeqnarray}
For $\sigma_{dB} = 0$ (no shadowing), let $\Psi_s(r)$ be the area of the disc centred at the location of the source $s$ with radius $d_r$ such that \eqref{eqn:c_s_t_r} holds ($\Psi_s(r)$ defined in the Theorem \ref{thm:average_throughput}). The probability of $k$ Poisson arrivals in the area $\Psi_s(r)$ can be obtained as  
\begin{IEEEeqnarray}{rCl}
\Pr\set{ |\mathcal{E}_s(r)| = k } = \frac{\left[\lambda \Psi_s(r)\right]^k}{k!}e^{-\lambda \Psi_s(r)}. \label{eqn:p_rs_greater_r}
\end{IEEEeqnarray}
The probability in \eqref{eqn:pr_k_greater_k_r} is either zero when $k<k_r$ or one when $k\ge k_r$. Therefore, the complementary CDF of the rate $R_s$ can be upper bounded as    
\begin{IEEEeqnarray}{rCl}
\Pr\set{R_s > r} &\le& \sum_{k\ge k_r}\Pr\set{|\mathcal{E}_s(r)| = k}\nonumber\\
&=& \sum_{k \ge k_r} \frac{\left[\lambda \Psi_s(r)\right]^k}{k!}e^{-\lambda \Psi_s(r)}. 
\end{IEEEeqnarray}
Since $R_s$ is a non-negative random variable, we can express the expected value of $R_s$ as 
\begin{IEEEeqnarray}{rCl}
E\left[ R_s \right] &=& \int_{0}^{\infty} \Pr\set{R_s > r} \textnormal{d}r\nonumber \\ 
&=& \int_{0}^{\C{\gamma}} \textnormal{d}r +  \int_{\C{\gamma}}^{\infty} \Pr\set{R_s> r} \textnormal{d}r  \nonumber \\ 
& \le &{\C{\gamma}} + \int_{\C{\gamma}} ^{\infty }   \sum_{k \ge k_r} \frac{\left[\lambda \Psi_s(r)\right]^k}{k!}e^{-\lambda \Psi_s(r)}    \textnormal{d}r. 
 \end{IEEEeqnarray}
%The quantity $\C{\gamma} = \C{\abs{h_s}^2{p_t}/{\pn}} $ can be interpreted as the threshold to determine whether a VMIMO should be formed or not. 
%\end{proof}

%{\color{magenta}

\section{Proof of Lemma \ref{thm:average_throughput_lns}} \label{sec:upper_bound_shadowing}

For the case with the log-normal shadowing ($\sigma_{dB} > 0$), we adopt an asymptotic  method to find  $\Pr\set{|\mathcal{E}_s| = k}$ in \eqref{eqn:p_rs_greater_r}.%\footnote{A different derivation approach is adopted in \cite{Orriss2003}.} %It is noteworthy that the closed form equation  given in \cite[Theorem 1]{Buratti2011} is not directly applicable to the case with shadowing since evaluating the probability of the connectivity in that analysis still requires an asymptotic approximation.

Assume that relays are randomly located according to a shrinking Bernoulli process on a disk with radios $d_{\max}$ centred at the source.
For this process, $2\pi \lambda \delta d_i$ is the probability that there exists one user in the distance interval of  $[d_i - \delta, d_i)$ of the source node where $d_i = i\delta$. 
This process asymptotically approaches a Poisson process with rate $\lambda$ as $\delta \rightarrow 0$~\cite{Gallager1996}.
Let $n_{\max} = \lfloor d_{\max} / \delta \rfloor$. For $i = 1, \cdots, n_{\max}$, let% the events $E_i$ and $E_i'$ be defined as 
\begin{IEEEeqnarray*}{rCl}
E_i &\stackrel{\Delta}{=}& \set{\text{exists one relay $l$ s.t. } d_{sl} e^{-\sigma V_l} \in [(i-1) \delta, i\delta)}, \\ 
E_i' &\stackrel{\Delta}{=}& \set{\text{exists one relay at } d_i }.
\end{IEEEeqnarray*}
Furthermore, by definition,  let 
\begin{IEEEeqnarray}{rCl}
\pji = \Pr\set{ j \delta e^{-\sigma V} \in [ (i-1) \delta, i\delta) | E_j'}
\end{IEEEeqnarray}
for $i, j \in \set{1, 2, \cdots, n_{\max}}$ where $V$ is a normal Gaussian random variable. 
%
%The probability $\pji$ can be intuitively interpreted as the probability that while a user is spatially located at distance $d(j)$ of the source, it is \emph{effectively} within the distance $[d(i) - \delta, d(i))$ of the source when it is affected by the shadowing. 
%
This probability can be computed as  
\begin{IEEEeqnarray}{rCl}
\pji & \stackrel{(a)}{=} & \Pr \set{ (i-1)\delta \le j\delta e^{\sigma V} < i\delta } \nonumber \\ 
%& = & \Pr \set{\frac{1}{\sigma} \ln \frac{i-1}{j} \le V < \frac{1}{\sigma} \ln \frac{i}{j}} \nonumber \\ 
& = & \frac{1}{2} \text{erf}\left(\frac{1}{\sqrt{2} \sigma} \ln \frac{i}{j} \right) - 
\frac{1}{2} \text{erf}\left(\frac{1}{\sqrt{2} \sigma} \ln \frac{i-1}{j} \right)
\end{IEEEeqnarray}
where $(a)$ holds since  $V$ and $-V$ have the same  distributions. %, and $\text{erf}(\cdot)$ is the Gauss error function. 
The probability of the event $E_i$, $p_i$, can be written as 
\begin{IEEEeqnarray}{rCl}
p_i 
%& = & \Pr\set{E_i} \nonumber\\ 
& =& \sum_{j=1}^{n_{\max}} \Pr\set{ j \delta e^{-\sigma V} \in [ (i-1) \delta, i\delta) | E_j'} \Pr\set{E_j'} \nonumber\\ 
%& =& \sum_{j=1}^{n_{\max}} \pji  2 \pi \delta d(j) \lambda  \nonumber \\
& =& 2 \pi \lambda \delta^2  \sum_{j=1}^{\lfloor d_{\max} / \delta \rfloor} j \pji.
\end{IEEEeqnarray}
Now, the event $\set{|\mathcal{E}_s(r)| = k}$ is tantamount to the event of $k$ successes out of $n$ YES/NO trials with probability of successes $p_1, p_2, \cdots, p_n$ which has a Poisson binomial distribution ($n = \lfloor d_r / \delta \rfloor$).
%There are closed form for the this probability distribution% \cite{chen1997 }
%; however, in this paper, a recursive computation method  \cite{shah1973} is adopted. 
%In this scenario the number of trials is the threshold distance $d_r$ in \eqref{eqn:threshold_distance} divided by the distance finite element $\delta$, i.e., 
% $n = \lfloor d_r / \delta \rfloor$.
% After evaluating  $\Pr\set{|\mathcal{E}_s(r)| = k}$, it can be substituted in \eqref{eqn:p_rs_greater_r}, and rest of calculations are similar to the proof of Theorem \ref{thm:average_throughput}. 
%
Let   $\Pr\set{|\mathcal{E}_s(r)| = k} = \lim_{\delta \rightarrow 0} \pi_k(r, \delta) $. Using a recursive method~\cite{shah1973}, we have
\begin{IEEEeqnarray}{rCl}\label{eqn:pi_k}
\pi_k (r, \delta) = \left\{ 
\begin{IEEEeqnarraybox}[][c]{ls} %\label{eqn:recursive_poisson_binomial}
%\IEEEstrut
\prod_{i=1}^{ n} (1-p_i),  &  $k = 0$,  \\ 
\frac{1}{k}\sum _{j = 1}^{k} (-1)^{j-1} \pi_{k - j}(r, \delta) T_j, \quad & $ k > 0$
\IEEEstrut
\end{IEEEeqnarraybox}
\right. 
\end{IEEEeqnarray}
where 
\begin{IEEEeqnarray*}{rCl}
T_j = \sum_{i=1}^{ n} \left(\frac{p_i}{1 - p_i} \right)^j. 
\end{IEEEeqnarray*}

%} %END OF COLOR MAGENTA

\section{Proof of Proposition \ref{prop:best_relay_set}}\label{sec:proof_prop_best_relay_set}
%\begin{proof}
%\begin{lem} \label{thm:increasing_r}
We first claim that in Algorithm \ref{alg:culstering_one_stue}, $r_{\min}$ is a non-increasing function of $\abs{\mathcal{A}_s}$. % In other words, in every step of the \textbf{for} loop, $r_{\text{\UE}}$ either remains the same or becomes smaller. 
%\end{lem}
%\textbf{Proof:}
This claim can be verified as follows: considering that the set of available \UEs~for assistance, $\mathcal{E}_s^{\text{sorted}}$, is sorted with respect to the link rates $R_{sl}$, $l \in \mathcal{E}_s^{\text{sorted}}$, each addition to the $\mathcal{A}_s$ either limits $r_{\min}$ to a lower value or leaves it intact. 
%\hfill $\blacksquare$
%\begin{lem}\label{thm:decreasing_gamma}

Moreover, we claim that in Algorithm \ref{alg:culstering_one_stue}, $\gamma_{\text{AP}}$ is a non-decreasing function of $\abs{\mathcal{A}_s}$ if the precoding resolution is high enough ($\nw = \infty$), 
%\end{lem}
%\textbf{Proof:} 
and this claim can be easily verified by observing that  in every step of the \textbf{for} loop, $\gamma_{\text{AP}}$ either remains the same or becomes larger.
%\hfill $\blacksquare$ 

Considering these two claims,  since $r_{\min}$ is increasing through the loop and $\gamma_{\text{AP}}$ is decreasing through the loop, there is an optimal point in the loop that adding or deleting any \UE~to or from  $\mathcal{A}_s$ would only decrease $r_{\text{new}}$.  
%\end{proof}

%================================================================================

%\nocite{*} %Even non-cited BibTeX-Entries will be shown.
%\bibliographystyle{alpha}
\bibliographystyle{IEEEtran} %Style of Bibliography:  plain / apalike / amsalpha / ...
\bibliography{vmimo} %,IT_MIMO}

\begin{IEEEbiography}{S. Hossein Seyedmehdi} received the B.Sc. degree in electrical engineering form Iran University of Science and Technology, Tehran, Iran in 2005 and the M.Eng. degree in electrical and computer engineering from National University of Singapore, Singapore, in 2008. He is currently pursuing the Ph.D. degree in electrical and computer engineering at University of Toronto, Toronto, Canada. 

In the summers of 2012 and 2013, he was a Research Intern at Ericsson, Ottawa, Canada where he worked on Smart Applications on Virtual Infrastructures (SAVI) project. His research interests include Information Theory for Wireless Communications, Signal Processing for MIMO channels, Algorithms for Radio Access Networks, and Future Radio Technologies. 
\end{IEEEbiography}

\begin{IEEEbiography}{Gary Boudreau (M'84--SM'11)}
received a B.A.Sc. in Electrical Engineering from the University of Ottawa in 1983, an M.A.Sc. in Electrical Engineering from Queens University in 1984 and a Ph.D. in electrical engineering from Carleton University in 1989. 

From 1984 to 1989 he was employed as a communications systems engineer with Canadian Astronautics Limited after which from 1990 to 1993 he worked as a satellite systems engineer for MPR Teltech Ltd. For the period spanning 1993 to 2009 he was employed by Nortel Networks in a variety of communication systems and management roles within the CDMA and LTE basestation product groups.  In 2010 he joined Ericsson Canada where he is currently working in the LTE systems architecture group. His interests include digital and wireless communications as well as digital signal processing.
\end{IEEEbiography}

\end{document}

%% file: fig_rate_vs_density_vs_SNR_dmax_25m.tikz
% This file was created by matlab2tikz v0.3.0.
% Copyright (c) 2008--2012, Nico Schlömer <nico.schloemer@gmail.com>
% All rights reserved.
% 
% The latest updates can be retrieved from
%   http://www.mathworks.com/matlabcentral/fileexchange/22022-matlab2tikz
% where you can also make suggestions and rate matlab2tikz.
% 
% 
% 
\begin{tikzpicture}

\begin{axis}[%
width=\figurewidth,
height=\figureheight,
scale only axis,
xmin=0, xmax=0.07,
xlabel={$\lambda\text{ (\UEs/m}^\text{2}\text{)}$},
xmajorgrids,
ymin=0, ymax=6,
ylabel={Average spectral efficiency (bps/Hz)},
ymajorgrids,
legend style={at={(0.74,0.01)},anchor=south west,draw=black,fill=white,align=left, legend cell align=left}]
\addplot [
color=red,
solid,
line width=0.6pt,
mark=x,
mark options={solid}, 
]
coordinates{
 (0,3.45632143459955)(0.002,3.51299640982874)(0.004,3.67804718779362)(0.008,4.02201770881534)(0.016,4.44432238303976)(0.032,4.86180349278019)(0.064,5.26551547614669) 
};
\addlegendentry{(i)};

\addplot [
color=black,
dashed,
line width=0.6pt,
mark=diamond,
mark size=2.7pt,
mark options={solid},
]
coordinates{
 (0,3.45632143459959)(0.002,3.51566686639805)(0.004,3.68016515459089)(0.008,4.04746101803633)(0.016,4.45456285714206)(0.032,4.85728636359096)(0.064,5.25621906619164) 
};
\addlegendentry{(ii)};

\addplot [
color=blue,
solid,
line width=0.6pt,
mark=+,
mark options={solid}
]
coordinates{
 (0,3.45632858383998)(0.001,3.45903884366495)(0.002,3.47275083148219)(0.004,3.53472265803414)(0.008,3.7154303332346)(0.016,4.02014427893831)(0.032,4.36774824758307)(0.064,4.72557573715245) 
};
\addlegendentry{(iii)};

\addplot [
color=black,
dashed,
line width=0.6pt,
mark=o,
mark options={solid}
]
coordinates{
 (0,3.45632143459955)(0.001,3.45632143459955)(0.002,3.46428681140471)(0.004,3.53012492474967)(0.008,3.73097544598021)(0.016,4.03935858142983)(0.032,4.34811495844381)(0.064,4.72678706334344) 
};
\addlegendentry{(iv)};

\addplot [
color=black,
solid,
line width=0.6pt,
]
coordinates{
 (0,3.45632143459955)(0.001,3.45677858045807)(0.002,3.46333801390447)(0.004,3.51215731030475)(0.008,3.62390097802763)(0.016,3.87355908752149)(0.032,4.22590030297882)(0.064,4.56275728056917) 
};
\addlegendentry{(v)};

\addplot [
color=black,
solid,
line width=0.6pt,
mark=asterisk,
mark options={solid}
]
coordinates{
 (0,3.45632143459955)(0.001,3.45632143459955)(0.002,3.45632143459955)(0.004,3.45990950645988)(0.008,3.49806577634385)(0.016,3.61339175955012)(0.032,3.86917415828344)(0.064,4.23092586425521) 
};
\addlegendentry{(vi)};

\addplot [
color=red,
solid,
line width=0.6pt,
mark=x,
mark options={solid},
forget plot
]
coordinates{
 (0,0.998290228580272)(0.002,1.59761615093003)(0.004,1.96651764972158)(0.008,2.38017841144834)(0.016,2.79820517587914)(0.032,3.2086498041308)(0.064,3.60924551156461) 
};
\addplot [
color=black,
dashed,
line width=0.6pt,
mark=diamond,
mark size=2.7pt,
mark options={solid},
forget plot
]
coordinates{
 (0,0.998290228580277)(0.002,1.5980920415554)(0.004,1.97309611854506)(0.008,2.38340897392526)(0.016,2.7943346862905)(0.032,3.20491211525281)(0.064,3.59962358440582) 
};
\addplot [
color=blue,
solid,
line width=0.6pt,
mark=+,
mark options={solid},
forget plot
]
coordinates{
 (0,0.998294156851906)(0.001,1.29625460169678)(0.002,1.49644996360367)(0.004,1.75942869147098)(0.008,2.06004319667567)(0.016,2.3835911785803)(0.032,2.72337345568935)(0.064,3.07393312594937) 
};
\addplot [
color=black,
dashed,
line width=0.6pt,
mark=o,
mark options={solid},
forget plot
]
coordinates{
 (0,0.998290228580273)(0.001,1.32828983140015)(0.002,1.52226417784458)(0.004,1.70233770070272)(0.008,2.05008804829632)(0.016,2.38195025397518)(0.032,2.70164369636571)(0.064,3.06732970020608) 
};

\addplot [
color=black,
solid,
line width=0.6pt,
forget plot
]
coordinates{
 (0,0.998290228580273)(0.001,1.22308877003176)(0.002,1.40186610496208)(0.004,1.67538071247966)(0.008,1.90873114649765)(0.016,2.22082266512958)(0.032,2.58374486418741)(0.064,2.87197390319264) 
};

\addplot [
color=black,
solid,
line width=0.6pt,
mark=asterisk,
mark options={solid},
forget plot
]
coordinates{
 (0,0.998290228580273)(0.001,1.10868040660353)(0.002,1.22049641073308)(0.004,1.43319997265488)(0.008,1.66864846710744)(0.016,1.92823047612496)(0.032,2.27495270425543)(0.064,2.60227924865318) 
};

\addplot [
color=red,
solid,
line width=0.6pt,
mark=x,
mark options={solid},
forget plot
]
coordinates{
 (0,0.137192806896962)(0.002,0.448657518723062)(0.004,0.658446496141548)(0.008,0.936629717040356)(0.016,1.2651742799388)(0.032,1.62194463429098)(0.064,1.99158594191754) 
};
\addplot [
color=black,
dashed,
line width=0.6pt,
mark=diamond,
mark size=2.7pt,
mark options={solid},
forget plot
]
coordinates{
 (0,0.137192806896962)(0.002,0.44318895607392)(0.004,0.665616287196455)(0.008,0.937759437968276)(0.016,1.2585993012417)(0.032,1.62848097735923)(0.064,1.98567217629974) 
};
\addplot [
color=blue,
solid,
line width=0.6pt,
mark=+,
mark options={solid},
forget plot
]
coordinates{
 (0,0.137193520436673)(0.001,0.295549911688937)(0.002,0.393300499959327)(0.004,0.530103623357532)(0.008,0.709957251853082)(0.016,0.935252481011385)(0.032,1.20182765961343)(0.064,1.50018070015414) 
};
\addplot [
color=black,
dashed,
line width=0.6pt,
mark=o,
mark options={solid},
forget plot
]
coordinates{
 (0,0.137192806896962)(0.001,0.295864541288705)(0.002,0.365363163131375)(0.004,0.529304302073799)(0.008,0.709556435181779)(0.016,0.93888254986573)(0.032,1.19019189840174)(0.064,1.46286620992847) 
};
\addplot [
color=black,
solid,
line width=0.6pt,
forget plot
]
coordinates{
 (0,0.137192806896962)(0.001,0.257865298973054)(0.002,0.351242967043475)(0.004,0.484295923773593)(0.008,0.652475354832825)(0.016,0.859646820791509)(0.032,1.08685508389605)(0.064,1.35500492239967) 
};

\addplot [
color=black,
solid,
line width=0.6pt,
mark=asterisk,
mark options={solid},
forget plot
]
coordinates{
 (0,0.137192806896962)(0.001,0.213431835685917)(0.002,0.269217747044534)(0.004,0.34373839393296)(0.008,0.482590849453234)(0.016,0.631575082802693)(0.032,0.809546567625871)(0.064,1.07090347997019) 
};

\end{axis}

	\draw[line width=3.5pt, draw=white] (1.2,1.45) arc  (50:320:0.2cm and 0.6cm) ;
	\draw[very thick] (1.2,1.45) arc  (50:320:0.2cm and 0.6cm) ; 
	\node [] at(2.2,0.4) {$\gamma = -10$ dB}; 
	
	\draw[line width=3.5pt, draw=white] (2,3.57) arc  (55:315:0.2cm and 0.82cm) ; 
	\draw[very thick] (2,3.57) arc  (55:315:0.2cm and 0.82cm) ; 
	\node [] at(2.8,2.3) {$\gamma = 0$ dB}; 
	
	\draw[line width=3.5pt, draw=white] (4,6) arc  (55:315:0.2cm and 0.9cm) ; 
	\draw[very thick] (4,6) arc  (55:315:0.2cm and 0.9cm) ; 
	\node [] at(4,6.5) {$\gamma = 10$ dB}; 
	
\end{tikzpicture}%

%% file: fig_SINR_imp_vs_baseline_SINR.tikz
% This file was created by matlab2tikz v0.3.0.
% Copyright (c) 2008--2012, Nico Schlömer <nico.schloemer@gmail.com>
% All rights reserved.
% 
% The latest updates can be retrieved from
%   http://www.mathworks.com/matlabcentral/fileexchange/22022-matlab2tikz
% where you can also make suggestions and rate matlab2tikz.
% 
% 
% 
\begin{tikzpicture}[scale = \figscale]

\begin{axis}[%
width=\figurewidth,
height=\figureheight,
scale only axis,
xmin=-15, xmax=15,
xlabel={$\text{SINR}_{\text{eff}}\text{ without VMIMO (dB)}$},
xmajorgrids,
ymin=-2, ymax=10,
ylabel={$\text{Improvement in SINR}_{\text{eff}}\text{ due to VMIMO (dB)}$},
ymajorgrids,
legend style={draw=black,fill=white,align=left, legend cell align=left}
 ]

\addplot [
color=black, %red,
dashed,
line width=0.75pt,
mark=triangle,
mark options={solid,,rotate=180}
]
coordinates{
 (-12.5,9.95392694468092)(-7.5,6.21096777387585)(-2.5,3.88083712043916)(2.5,1.53909697601485)(7.5,-0.285384029996687)(12.5,-0.942378674346861) 
};
\addlegendentry{{\footnotesize $\lambda= 0.064, \nw = 2$}};

\addplot [ 
color=black,
solid,
line width=0.75pt,
mark=triangle,
mark options={solid,,rotate=180}
]
coordinates{
 (-12.5,8.76613571727882)(-7.5,5.29136508457909)(-2.5,3.0420437709835)(2.5,0.873147726036997)(7.5,-0.434377717864089)(12.5,-0.73433737459155) 
};
\addlegendentry{{\footnotesize $\lambda= 0.064, \nw = 1$}};

\addplot [
color=brown, %red,
dashed,
line width=0.75pt,
mark=diamond,
mark options={solid}
]
coordinates{
 (-12.5,7.46407588523477)(-7.5,5.23962619120684)(-2.5,2.79703885222174)(2.5,0.793282733068538)(7.5,-0.314289095161701)(12.5,-0.563403829085464) 
};
\addlegendentry{{\footnotesize $\lambda= 0.032, \nw = 2$}};

\addplot [
color=brown,
solid,
line width=0.75pt,
mark=diamond,
mark options={solid}
]
coordinates{
 (-12.5,6.38673842803755)(-7.5,4.22973163467175)(-2.5,2.15278149920996)(2.5,0.390163738022414)(7.5,-0.347563275373387)(12.5,-0.653631729571621) 
};
\addlegendentry{{\footnotesize $\lambda= 0.032, \nw =1$}};

% -------------------------
\begin{comment}
\addplot [
color=black, %red,
dashed,
line width=0.75pt,
mark=+,
mark options={solid}
]
coordinates{
 (-12.5,4.85722409183132)(-7.5,2.59386798311071)(-2.5,1.20804587081769)(2.5,0.110183265531719)(7.5,-0.273975715987227)(12.5,-0.382816920830223) 
};
\addlegendentry{{\footnotesize $\lambda= 0.008, \nw = 2$}};

\addplot [
color=black,
solid,
line width=0.75pt,
mark=+,
mark options={solid}
]
coordinates{
 (-12.5,4.06567300350826)(-7.5,2.05996170656914)(-2.5,0.808030352967594)(2.5,-0.0469360729235399)(7.5,-0.267598370404165)(12.5,-0.338051592377421) 
};
\addlegendentry{{\footnotesize $\lambda= 0.008, \nw = 1$}};
\end{comment}
% -------------------------

\addplot [
color=red, %red,
dashed,
line width=0.75pt,
mark=o,
mark options={solid}
]
coordinates{
 (-12.5,2.38994414758407)(-7.5,1.66717123254244)(-2.5,0.654623334886874)(2.5,0.0794678324778697)(7.5,-0.15030964504086)(12.5,-0.250877725917067) 
};
\addlegendentry{{\footnotesize $\lambda= 0.004, \nw = 2$}};

\addplot [
color=red,
solid,
line width=0.75pt,
mark=o,
mark options={solid}
]
coordinates{
 (-12.5,1.95627651535519)(-7.5,1.28746911604293)(-2.5,0.42699026684962)(2.5,0.0238555336503212)(7.5,-0.133102841552482)(12.5,-0.222199762252439) 
};
\addlegendentry{{\footnotesize $\lambda= 0.004, \nw = 1$}};

\addplot [
color=red, %blue,
solid,
line width=0.75pt,
mark=triangle,
mark options={solid}
]
coordinates{
 (-12.5,2.00461141185634)(-7.5,1.33040988576604)(-2.5,0.459553859482554)(2.5,0.0529452866615073)(7.5,-0.0830218792329006)(12.5,-0.15860472390412) 
};
\addlegendentry{{\footnotesize $\lambda= 0.004$, exh.}};

\addplot [
color=blue, %red,
dashed,
line width=0.75pt,
mark options={solid}
]
coordinates{
 (-12.5,1.69658518531543)(-7.5,0.72024563572724)(-2.5,0.308387722296323)(2.5,0.017601744058569)(7.5,-0.0360771283446286)(12.5,-0.0552208515620954) 
};
\addlegendentry{{\footnotesize $\lambda= 0.002, \nw = 2$}};

\addplot [
color=blue,
dotted,
line width=0.75pt,
mark=+, %asterisk,
mark options={solid}
]
coordinates{
 (-12.5,1.38270521670568)(-7.5,0.514745753101735)(-2.5,0.198004687533624)(2.5,0.00936527550668523)(7.5,-0.0438566665456568)(12.5,-0.0598814421555179) 
};
\addlegendentry{{\footnotesize $\lambda= 0.002, \nw = 1$}};

\addplot [
color=blue,
solid,
line width=0.75pt,
mark options={solid}
]
coordinates{
 (-12.5,1.38754708378738)(-7.5,0.542103760697368)(-2.5,0.20308709687808)(2.5,0.015774268896006)(7.5,-0.0318930716533728)(12.5,-0.0454999853896487) 
};
\addlegendentry{{\footnotesize $\lambda= 0.002$, exh.}};

\end{axis}
\end{tikzpicture}%

%% file: fig_num_relays_vs_baseline_SINR.tikz
% This file was created by matlab2tikz v0.3.0.
% Copyright (c) 2008--2012, Nico Schlömer <nico.schloemer@gmail.com>
% All rights reserved.
% 
% The latest updates can be retrieved from
%   http://www.mathworks.com/matlabcentral/fileexchange/22022-matlab2tikz
% where you can also make suggestions and rate matlab2tikz.
% 
% 
% 

% defining custom colors
\definecolor{mycolor1}{rgb}{0.847058832645416,0.160784319043159,0}

\begin{tikzpicture} [scale=\figscale]

\begin{axis}[%
width=\figurewidth,
height=\figureheight,
scale only axis,
xmin=-15, xmax=15,
xlabel={$\text{SINR}_{\text{eff}}\text{ without VMIMO (dB)}$},
xmajorgrids,
ymin=0, ymax=5,
ylabel={Average no. of assisting relays},
ymajorgrids,
legend style={draw=black,fill=white,align=left, legend cell align=left}]

\addplot [
color=black, %mycolor1,
dashed,
line width=0.75pt,
mark=triangle,
mark options={solid,,rotate=180}
]
coordinates{
 (-12.5,4.60655737704918)(-7.5,3.63544303797468)(-2.5,3.22249388753056)(2.5,2.47356948228883)(7.5,0.901382488479263)(12.5,0.182370820668693) 
};
\addlegendentry{{\footnotesize $\lambda\text= 0.064,\nw= 2$}};

\addplot [
color=black,
solid,
line width=0.75pt,
mark=triangle,
mark options={solid,,rotate=180}
]
coordinates{
 (-12.5,4.40983606557377)(-7.5,3.5873417721519)(-2.5,3.32273838630807)(2.5,2.19618528610354)(7.5,0.548387096774194)(12.5,0) 
};
\addlegendentry{{\footnotesize $\lambda\text= 0.064,\nw= 1$}};

\addplot [
color= brown, %black, %red,
dashed,
line width=0.75pt,
mark=diamond,
mark options={solid}
]
coordinates{
 (-12.5,3.12 )(-7.5,2.91 )(-2.5,2.14 )(2.5,1.31 )(7.5,0.45 )(12.5,0 ) 
};
\addlegendentry{{\footnotesize $\lambda\text= 0.032,\nw= 2$}};

\addplot [
color=brown, %black,
solid,
line width=0.75pt,
mark=diamond,
mark options={solid}
]
coordinates{
 (-12.5,3.08 )(-7.5,2.9 )(-2.5,2.06 )(2.5,1.06 )(7.5,0.3 )(12.5,0) 
};
\addlegendentry{{\footnotesize $\lambda\text= 0.032,\nw= 1$}};

% -------------------------
\begin{comment}
\addplot [
color=black, %red,
dashed,
line width=0.75pt,
mark=+,
mark options={solid}
]
coordinates{
 (-12.5,1.8)(-7.5,1.32136752136752)(-2.5,0.93298969072165)(2.5,0.404371584699454)(7.5,0.062575210589651)(12.5,0) 
};
\addlegendentry{{\footnotesize $\lambda\text= 0.008,\nw= 2$}};

\addplot [
color=black,
solid,
line width=0.75pt,
mark=+,
mark options={solid}
]
coordinates{
 (-12.5,1.74782608695652)(-7.5,1.31452991452991)(-2.5,0.851988217967599)(2.5,0.281724347298118)(7.5,0.02647412755716)(12.5,0) 
};
\addlegendentry{{\footnotesize $\lambda\text= 0.008,\nw= 1$}};
\end{comment}
% -------------------------

\addplot [
color=red, %red,
dashed,
line width=0.75pt,
mark=o,
mark options={solid}
]
coordinates{
 (-12.5,0.727272727272727)(-7.5,0.796762589928058)(-2.5,0.485496183206107)(2.5,0.213224368499257)(7.5,0.0269360269360269)(12.5,0) 
};
\addlegendentry{{\footnotesize $\lambda\text= 0.004,\nw= 2$}};

\addplot [
color=red,
solid,
line width=0.75pt,
mark=o,
mark options={solid}
]
coordinates{
 (-12.5,0.727272727272727)(-7.5,0.776978417266187)(-2.5,0.393893129770992)(2.5,0.135215453194651)(7.5,0.0134680134680135)(12.5,0) 
};
\addlegendentry{{\footnotesize $\lambda\text= 0.004,\nw= 1$}};

\addplot [
color=red, %blue,
solid,
line width=0.75pt,
mark=triangle,
mark options={solid}
]
coordinates{
 (-12.5,0.648484848484848)(-7.5,0.681654676258993)(-2.5,0.329007633587786)(2.5,0.0936106983655275)(7.5,0.0134680134680135)(12.5,0) 
};
\addlegendentry{{\footnotesize $\lambda\text= 0.004$, exh.}};

\addplot [
color=blue, %red,
dashed,
line width=0.75pt,
mark options={solid}
]
coordinates{
 (-12.5,0.634920634920635)(-7.5,0.364312267657993)(-2.5,0.221742881794651)(2.5,0.0596658711217184)(7.5,0.011611030478955)(12.5,0) 
};
\addlegendentry{{\footnotesize $\lambda\text= 0.002,\nw= 2$}};

\addplot [
color=blue,
dotted,
line width=0.75pt,
mark=+, %asterisk,
mark options={solid}
]
coordinates{
 (-12.5,0.619047619047619)(-7.5,0.349442379182156)(-2.5,0.181190681622088)(2.5,0.0421638822593477)(7.5,0.00870827285921625)(12.5,0) 
};
\addlegendentry{{\footnotesize $\lambda\text= 0.002,\nw= 1$}};

\addplot [
color=blue,
solid,
line width=0.75pt,
mark options={solid}
]
coordinates{
 (-12.5,0.587301587301587)(-7.5,0.308550185873606)(-2.5,0.160483175150992)(2.5,0.0350039777247414)(7.5,0.00435413642960813)(12.5,0) 
};
\addlegendentry{{\footnotesize $\lambda\text= 0.002$, exh.}};

\end{axis}
\end{tikzpicture}%

%% file: fig_SINR_imp_CDF.tikz
% This file was created by matlab2tikz v0.3.0.
% Copyright (c) 2008--2012, Nico Schlömer <nico.schloemer@gmail.com>
% All rights reserved.
% 
% The latest updates can be retrieved from
%   http://www.mathworks.com/matlabcentral/fileexchange/22022-matlab2tikz
% where you can also make suggestions and rate matlab2tikz.
% 
% 
% 
\begin{tikzpicture} [scale=\figscale]

\begin{axis}[%
width=\figurewidth,
height=\figureheight,
scale only axis,
xmin=-20, xmax=20,
xlabel={$\text{SINR}_{\text{eff}}$ (dB)},
xmajorgrids,
ymin=0, ymax=1,
ylabel={CDF},
ymajorgrids,
legend style={at={(0.0253298317042467,0.582103823915785)},anchor=south west,draw=black,fill=white,align=left, legend cell align=left}]

\addplot [
color=blue, %black,
solid
]
coordinates{
 (-19.6790813564043,0.000414078674948248)(-19.6790813564043,0.000414078674948248)(-19.6790813564043,0.000414078674948248)(-19.6790813564043,0.000414078674948248)(-19.6790813564043,0.000414078674948248)(-19.6790813564043,0.000414078674948248)(-18.7098124901681,0.00103519668737062)(-17.7587999674524,0.00186335403726712)(-16.4618275027091,0.00289855072463763)(-15.5582522840298,0.00393374741200825)(-14.8762070990048,0.00517598343685299)(-13.7692747629723,0.00724637681159424)(-12.8818992905654,0.00890269151138723)(-11.8693798912473,0.0113871635610767)(-10.7744947492817,0.0165631469979298)(-9.89661570257961,0.0246376811594206)(-8.85371170869997,0.0362318840579713)(-7.87881581795437,0.0496894409937891)(-6.90145347095962,0.0718426501035204)(-5.90320071387266,0.0993788819875789)(-4.9001088684313,0.131677018633542)(-3.90139084851273,0.170600414078678)(-2.89994745126509,0.218012422360252)(-1.90345585925901,0.272463768115947)(-0.901895671867786,0.33540372670808)(0.0959873345204055,0.403933747412016)(1.09919052060203,0.47391304347827)(2.09708194970935,0.56190476190477)(3.09600769338047,0.632712215320919)(4.09656656649547,0.695859213250527)(5.09931544174,0.749482401656325)(6.10529501496741,0.800000000000009)(7.09583117440881,0.84741200828158)(8.09817512666348,0.881780538302283)(9.1070951313014,0.91035196687371)(10.1005088436213,0.934368530020707)(11.0998867873856,0.956521739130436)(12.1105379513326,0.97329192546584)(13.2058454751703,0.983022774327123)(14.1078186357618,0.990062111801243)(15.1010578717897,0.993995859213251)(16.1676955397333,0.994824016563147)(17.1223752754401,0.997101449275362)(18.109188095327,0.997929606625259)(19.6021182009147,0.999171842650104)(20.5047941179591,0.999585921325052) 
};
\addlegendentry{No VMIMO};

\addplot [
color=red, %black,
solid,
mark=x,
mark options={solid}
]
coordinates{
 (-17.7187161452914,0.000414078674948248)(-14.6071604450445,0.000621118012422373)(-14.6071604450445,0.000621118012422373)(-14.6071604450445,0.000621118012422373)(-13.630045032456,0.00165631469979299)(-12.533140763881,0.00289855072463763)(-11.7129911567032,0.00517598343685299)(-10.6431700682687,0.00662525879917186)(-9.66277717155507,0.0105590062111802)(-8.70426410131065,0.0157349896480333)(-7.69309074232756,0.0240165631469982)(-6.71139735954116,0.0362318840579713)(-5.70589240967533,0.0569358178053835)(-4.71591063186889,0.0848861283643902)(-3.71501544329312,0.116770186335405)(-2.71894926594582,0.159627329192549)(-1.71867363058609,0.223809523809528)(-0.718103335126657,0.293581780538307)(0.2852023114544,0.373498964803319)(1.28166928367629,0.4608695652174)(2.28119884241485,0.554244306418228)(3.28143766590388,0.642650103519677)(4.28507783039794,0.709523809523819)(5.28275430191005,0.768530020703944)(6.28131540788031,0.820082815734997)(7.28875045124386,0.861283643892346)(8.2975653672221,0.896687370600418)(9.28239261856956,0.922360248447208)(10.2844303450546,0.945548654244309)(11.2827603836511,0.96376811594203)(12.2945455070892,0.978053830227744)(13.2829853315553,0.986542443064183)(14.3566638830258,0.992132505175984)(15.3451918911086,0.994409937888199)(16.3545741035007,0.996066252587992)(17.4792721651058,0.998343685300207)(18.3015873621218,0.999171842650104)(20.4826443245406,0.999792960662526) 
};
\addlegendentry{$\lambda\text{ = 0.008,\nw = 1}$};

\addplot [
color=red, %black,
dashed,
line width=0.8pt
]
coordinates{
 (-17.6683275297894,0.000414078674948248)(-14.462616948068,0.000621118012422373)(-14.462616948068,0.000621118012422373)(-14.462616948068,0.000621118012422373)(-13.5526386733733,0.00103519668737062)(-12.5659768912788,0.00207039337474124)(-11.4398913695596,0.00310559006211175)(-10.5531754249349,0.00414078674948237)(-9.67369161625564,0.00786749482401661)(-8.70426410131065,0.0128364389233956)(-7.60441265012564,0.0204968944099382)(-6.71316460801245,0.0312629399585924)(-5.70736476339244,0.0486542443064185)(-4.71723783504476,0.0759834368530029)(-3.70582284498445,0.10393374741201)(-2.71894926594582,0.139751552795033)(-1.71413922452112,0.199171842650107)(-0.718103335126657,0.265010351966878)(0.284102175573067,0.34140786749483)(1.28166928367629,0.431884057971022)(2.28178753809077,0.528571428571437)(3.28143766590388,0.623809523809532)(4.28461163991726,0.69648033126295)(5.28446483386606,0.758178053830238)(6.28131540788031,0.816977225672886)(7.28875045124386,0.860041407867501)(8.28678926827715,0.894202898550729)(9.28343250711874,0.924016563147001)(10.2850480072194,0.944927536231886)(11.2827603836511,0.964596273291927)(12.2917500607511,0.97784679089027)(13.2829853315553,0.986542443064183)(14.3566638830258,0.992132505175984)(15.3451918911086,0.994616977225673)(16.2922324970476,0.996687370600414)(17.377809913134,0.997929606625259)(18.2820987294735,0.998964803312629)(19.7420856852604,0.999585921325052)(21.0170797168994,0.999792960662526) 
};
\addlegendentry{$\lambda\text{ = 0.008,\nw = 2}$};

\addplot [
color=black,
solid,
mark=o,
mark options={solid}
]
coordinates{
 (-9.60665525578214,0.002)(-8.85373111450471,0.00240000000000007)(-7.8910055119256,0.00500000000000012)(-6.73536589727406,0.00860000000000005)(-6.27541179234405,0.0106000000000002)(-5.20362109770015,0.0204000000000004)(-4.32746598174164,0.0314000000000008)(-3.3230772998542,0.0486000000000008)(-2.32681649982521,0.0724000000000001)(-1.34738173818529,0.1236)(-0.331627745095437,0.192)(0.659054134628655,0.2628)(1.66124874035564,0.3646)(2.64574647481553,0.488600000000001)(3.64454754409295,0.583000000000001)(4.64330076219963,0.681800000000001)(5.64449109044171,0.762200000000001)(6.64400755380133,0.823000000000001)(7.64126803008454,0.868200000000001)(8.65452510903526,0.9224)(9.6694270532354,0.9388)(10.6679113540902,0.9524)(11.6438895248536,0.9614)(12.6790940251988,0.9794)(13.743872950647,0.986)(15.0002246353473,0.993)(15.7668718918578,0.9956)(17.0397160248816,0.999)(18.3154639419551,1) 
};
\addlegendentry{$\lambda\text{ = 0.064,\nw = 1}$};

\addplot [
color=black,
dash pattern=on 1pt off 3pt on 3pt off 3pt,
line width=1.0pt
]
coordinates{
 (-8.55585267478095,0.00139999999999996)(-8.55585267478095,0.00139999999999996)(-7.87643366231423,0.00299999999999989)(-6.77408015015315,0.00459999999999983)(-6.34780427571227,0.00759999999999983)(-5.3060244172307,0.0143999999999999)(-4.34078894168139,0.0231999999999998)(-3.30104118094701,0.0365999999999997)(-2.34351702863101,0.0551999999999994)(-1.35632572574433,0.0909999999999994)(-0.333530903262936,0.134)(0.641692058373366,0.2062)(1.64274771266709,0.291000000000001)(2.64161371124821,0.395400000000001)(3.67504287543364,0.512)(4.64725864902734,0.6208)(5.68001422927066,0.7244)(6.64486308560882,0.795800000000001)(7.64524485714052,0.8566)(8.64637122069513,0.906)(9.71035536816049,0.939)(10.657276512389,0.954)(11.769291884899,0.965)(12.6780602750015,0.9814)(13.947669865172,0.989)(15.0002246353473,0.993)(15.7668718918578,0.9966)(17.3231105545621,0.999)(18.3014374957463,1) 
};
\addlegendentry{$\lambda\text{ = 0.064,\nw = 2}$};

\end{axis}
\end{tikzpicture}%